\documentclass[11pt]{article}
\usepackage{sectsty}

\usepackage{graphicx}

\usepackage[colorlinks=true]{hyperref}
\hypersetup{urlcolor=blue, citecolor=red, linkcolor=blue}
\usepackage{amsmath,amssymb}
\usepackage{bm}
\usepackage{enumerate}
\usepackage{amsthm}
\usepackage[all,2cell]{xy}
\usepackage{mathrsfs}
\usepackage{mathtools}
\usepackage{tikz-cd}

\usepackage{caption}
\usepackage{subcaption}

\usetikzlibrary{cd}

\usepackage[a4paper, left=20mm, right=20mm, top=25mm, bottom=25mm]{geometry}

\usepackage{marginnote}

\usepackage{imakeidx}
\makeindex[intoc]

\usepackage[nottoc]{tocbibind}

\newtheorem{thm}{Theorem}[section]
\newtheorem{dfn}[thm]{Definition}
\newtheorem{prop}[thm]{Proposition}
\newtheorem{cor}[thm]{Corollary}

\newtheorem{exmpl}[thm]{Example}
\newtheorem{rmrk}[thm]{Remark}

\newcommand\restr[2]{{
  \left.\kern-\nulldelimiterspace 
  #1 
  \right|_{#2} 
}}

\newcommand*{\transp}[2][-3mu]{\ensuremath{\mskip1mu\prescript{\smash{\mathrm t\mkern#1}}{}{\mathstrut#2}}}%

\newcommand{\R}{\mathbb{R}}
\renewcommand{\d}{\mathrm{d}}

\newcommand{\Cinfty}{\mathscr{C}^\infty}
\newcommand{\T}{\mathrm{T}}
\newcommand{\cT}{\mathrm{T}^\ast}

\newcommand{\Lie}{\mathscr{L}}

\newcommand{\X}{\mathfrak{X}}

\renewcommand{\L}{\mathcal{L}}
\newcommand{\F}{\mathcal{F}}

\renewcommand{\P}{\mathcal{P}}

\newcommand{\Reeb}{\mathcal{R}}

\newcommand{\C}{\mathcal{C}}

\newcommand{\cl}{\mathfrak{cl}}

\newcommand{\parder}[2]{\frac{\partial #1}{\partial #2}}
\newcommand{\dparder}[2]{\dfrac{\partial #1}{\partial #2}}
\newcommand{\tparder}[2]{\partial #1/\partial #2}
\newcommand{\parderr}[3]{\frac{\partial^2 #1}{\partial #2\partial #3}}

\DeclareMathOperator{\Ima}{Im}

\DeclareMathOperator{\supp}{Supp}

\DeclareMathOperator{\rk}{rank}

\DeclareMathAlphabet{\mathpzc}{OT1}{pzc}{m}{it}
\def\d{\mathrm{d}}

\title{{\sffamily Time-dependent contact mechanics}}

\author{{\sffamily 
$^a$Manuel de León%
\thanks{e-mail:
   mdeleon@icmat.es \ ORCID: 0000-0002-8028-2348}\ ,\
$^b$Jordi Gaset%
\thanks{e-mail:
   jordi.gaset@unir.net \ ORCID: 0000-0001-8796-3149}\ ,
$^c$Xavier Gr\`acia%
\thanks{e-mail:
   xavier.gracia@upc.edu \ ORCID: 0000-0003-1006-4086}\ ,}\\
{\sffamily
$^c$Miguel C. Muñoz-Lecanda%
\thanks{e-mail:
   miguel.carlos.munoz@upc.edu \ ORCID: 0000-0002-7037-0248}\ ,\
$^b$Xavier Rivas%
\thanks{e-mail:
   xavier.rivas@unir.net \ ORCID: 0000-0002-4175-5157}\ ,
}
\\[1ex]
\normalsize\itshape\sffamily
$^a$Instituto de Ciencias Matem\'aticas,
Consejo Superior de Investigaciones Cient\'ificas\\
\normalsize\itshape\sffamily 
and Real Academia de Ciencias, Madrid, Spain.
\\[1ex]
\normalsize\itshape\sffamily
$^b$Escuela Superior de Ingenier\'{\i}a y Tecnolog\'{\i}a,\\
\normalsize\itshape\sffamily
Universidad Internacional de La Rioja, Logro\~no, Spain.
\\[1ex]
\normalsize\itshape\sffamily
$^c$Department of Mathematics,
Universitat Polit\`ecnica de Catalunya,
Barcelona, Spain.
\\[1ex]
}
\date{{\sffamily \today}}

\begin{document}

\maketitle

\begin{abstract}\noindent
    Contact geometry allows us to describe some thermodynamic and dissipative systems. In this paper we introduce a new geometric structure in order to describe time-dependent contact systems: cocontact manifolds. Within this setting we develop the Hamiltonian and Lagrangian formalisms, both in the regular and singular cases. In the singular case, we present a constraint algorithm aiming to find a submanifold where solutions exist. As a particular case we study contact systems with holonomic time-dependent constraints. Some regular and singular examples are analyzed, along with numerical simulations.
\end{abstract}

\noindent\textbf{Keywords:} contact structure, time-dependent system, Hamiltonian system, dissipation, singular Lagrangian, holonomic constraints, Jacobi structure

\noindent\textbf{MSC\,2020 codes:}
37J55; 
70H03, 
70H05, 
53D05, 
53D10, 
53Z05, 
70H45 

\newpage

{\setcounter{tocdepth}{2}
\def\baselinestretch{1}
\small
\def\addvspace#1{\vskip 1pt}
\parskip 0pt plus 0.1mm
\tableofcontents
}

\section{Introduction}

In the last decades, the interest in the formalism and applications of differential geometric structures to the study of mathematical physics and dynamical systems has risen drastically \cite{Abr1978,Arn1989,Lib1987,Gia1997}. The natural geometric framework for autonomous conservative mechanical systems, both Hamiltonian and Lagrangian, is symplectic geometry and the variational approach is based on Hamilton's variational principle. For non-autonomous systems \cite{Car1993,DeLeo2002,Ech1991,Kru1997,Mas2003} the same variational principle is valid, but the underlying geometry is cosymplectic geometry \cite{Alb1989,Can1992,Chi1994}. Time-dependent systems can also be geometrically formulated as the Reeb dynamics of a contact system. Both symplectic and cosymplectic structures are particular instances of Poisson geometry \cite{DeLeo1989,Lib1987,Vaisman1980}.

Recently, the application of contact geometry to the study of dynamical systems has grown significantly \cite{Bra2017a,DeLeo2019b,Gas2019}. This is due to the fact that contact geometry is not only a way to geometrically model the time-dependency in mechanical systems \cite{DeLeo2017}, which can also be done by means of cosymplectic geometry, but it also allows us to describe mechanical systems with certain types of damping, quantum mechanics \cite{Cia2018}, circuit theory \cite{Got2016},  control theory \cite{Ram2017} and thermodynamics \cite{Bra2018,Sim2020}, among many others \cite{Liu2018}. In recent papers \cite{Ver2019,Bra2020b}, the authors consider time-dependent systems, although a rigorous geometric description of time-dependent contact mechanics is in order. These systems include mechanical systems with time-dependent external forces and, in particular, controlled systems whose controls can be used to compensate the damping.

Let us recall that the variational approach to contact systems is based on Herglotz's principle \cite{Her1930} (see for instance \cite{DeLeo2021,DeLeo2021c,DeLeo2019c}). This variational principle generalizes the Hamilton principle and it is appropriate to deal with action-dependent Lagrangians. From the geometric perspective, it is important to point out that contact structures are not Poisson, but Jacobi. This is due to the fact that the Leibniz rule is not fulfilled \cite{DeLeo2019b}.

In the present paper, we develop a formalism in order to geometrically describe contact mechanical systems with explicit time-dependence. We begin by introducing a new geometric structure: cocontact manifolds. These manifolds extend the notion of both contact and cosymplectic structures. We study the properties of these manifolds, exhibit some examples, prove a Darboux-type theorem for these structures, and show that cocontact manifolds are Jacobi manifolds. In addition, we introduce the notion of cocontact orthogonal complement of a submanifold and define and characterize the coisotropic and Legendrian submanifolds of a cocontact manifold. This geometric framework is later used to develop both Hamiltonian and Lagrangian formulations of time-dependent mechanical systems with dissipation.

The introduction of time-dependence makes us consider systems subjected to constraints that can vary with time. Such systems can be described by introducing the constraints in the Lagrangian via Lagrange multipliers. These Lagrangians are obviously singular, since the velocity associated to the Lagrange multipliers do not appear in the expression of the Lagrangian. In order to study such systems, we introduce the notion of precocontact structure as a weakened version of a cocontact structure, and describe a suitable constraint algorithm.

The structure of the paper is as follows. In Section 2, we describe the geometrical setting that will be used throughout the paper. We present the notion of cocontact manifold and some examples. In particular, we see that cocontact manifolds are Jacobi manifolds, and thus all the theory about Jacobi structures is applicable to cocontact structures. Section 3 is devoted to develop a Hamiltonian formulation for dissipative time-dependent systems, while Section 4 dedicated to develop the Lagrangian formulation for cocontact systems and state the associated generalized Herglotz--Euler--Lagrange equations. In Section 5 we study the case of mechanical systems described by singular time-dependent contact Lagrangians and give a description of the constraint algorithm to obtain the dynamics in both the Hamiltonian and Lagrangian formulations. Section 6 is devoted to analyze the case of damped mechanical systems with holonomic constraints which can depend on time. Finally, in Section 7, three examples are presented and worked out: the damped forced harmonic oscillator, a system with time-dependent mass subjected to a central force with friction, and the damped pendulum with variable length. We also show and discuss some simulations of these systems.

Throughout the paper all the manifolds and mappings are assumed to be smooth and second-countable. Sum over crossed repeated indices is understood.

\section{Geometrical setting}

\subsection{Cocontact manifolds}

\begin{dfn}\label{dfn:cocontact-manifold}
	Let $M$ be a manifold of dimension $2n+2$. A \textbf{cocontact structure} on $M$ is a couple $(\tau,\eta)$ of 1-forms on $M$ such that $\tau$ is closed and such that $\tau\wedge\eta\wedge(\d\eta)^{\wedge n}$ is a volume form on $M$. Under these hypotheses, $(M,\tau,\eta)$ is called a \textbf{cocontact manifold}.
\end{dfn}

We can see from this definition that every cocontact manifold has two tangent distributions. The first one is generated by $\ker\tau$, is integrable, and gives a foliation made of contact leaves. The other one is $\ker\eta$ and it is not integrable. This structure will be very useful when proving the Darboux theorem.

\begin{exmpl}\rm
	Let $(P,\eta_0)$ be a contact manifold and consider the product manifold $M = \R\times P$. Denoting by $\d t$ the pullback to $M$ of the volume form in $\R$ and denoting by $\eta$ the pullback of $\eta_0$ to $M$, we have that $(\d t, \eta)$ is a cocontact structure on~$M$.
\end{exmpl}

\begin{exmpl}\rm
	Let $(P,\tau,-\d\theta)$ be an exact cosymplectic manifold and consider the product manifold $M = P\times\R$. Denoting by $s$ the coordinate in $\R$ we define the 1-form $\eta = \d s - \theta$. Then, $(\tau, \eta)$ is a cocontact structure on $M = P\times\R$.
\end{exmpl}

\begin{exmpl}[Canonical cocontact manifold]\label{ex:canonical-cocontact-manifold}\rm
	Let $Q$ be an $n$-dimensional smooth manifold with local coordinates $(q^i)$ and its cotangent bundle $\cT Q$ with induced natural coordinates $(q^i, p_i)$. Consider the product manifolds $\R\times\cT Q$ with coordinates $(t, q^i, p_i)$, $\cT Q\times\R$ with coordinates $(q^i, p_i, s)$ and $\R\times\cT Q\times\R$ with coordinates $(t, q^i, p_i, s)$ and the canonical projections
	\begin{center}
		\begin{tikzcd}
			& \R\times\cT Q\times\R \arrow[dl, swap, "\rho_1"] \arrow[dr, "\rho_2"] \arrow[dd, "\pi"] & \\
			\R\times\cT Q \arrow[dr, swap, "\pi_2"] & & \cT Q\times\R \arrow[dl, "\pi_1"] \\
			& \cT Q &
		\end{tikzcd}
	\end{center}
	Let $\theta_0\in\Omega^1(\cT Q)$ be the canonical 1-form of the cotangent bundle, which has local expression $\theta_0 = p_i\d q^i$.
	Denoting by $\theta_2 = \pi_2^\ast\theta$, we have that $(\d t, \theta_2)$ is a cosymplectic structure in $\R\times\cT Q$.
	
	On the other hand, denoting by $\theta_1 = \pi_1^\ast\theta_0$, we have that $\eta_1 = \d s - \theta_1$ is a contact form in $\cT Q\times\R$.
	
	Finally, consider the 1-form $\theta = \rho_1^\ast\theta_2 = \rho_2^\ast\theta_1 = \pi^\ast\theta_0\in\Omega^1(\R\times\cT Q\times\R)$ and let $\eta = \d s - \theta$. Then, $(\d t, \eta)$ is a cocontact structure in $\R\times\cT Q\times\R$. The local expression of the 1-form $\eta$ is
	$$ \eta = \d s - p_i\d q^i\,. $$
\end{exmpl}

\begin{prop}
Let $(M,\tau,\eta)$ be a cocontact manifold. We have the following isomorphism of vector bundles:
$$
	\begin{array}{rccl}
		\flat\colon & \T M & \longrightarrow & \cT M \\
		& v & \longmapsto & (i(v)\tau)\tau + i(v)\d\eta + \left(i(v)\eta\right)\eta
	\end{array}
$$
\end{prop}
\begin{proof}
	It is clear that $\ker\flat = 0$, since $M$ is a cocontact manifold. Hence, it follows that $\flat$ has to be an isomorphism.
\end{proof}

This isomorphism can be extended to an isomorphism of $\Cinfty(M)$-modules:
$$
	\begin{array}{rccl}
		\flat\colon & \X(M) & \longrightarrow & \Omega^1(M) \\
		& X & \longmapsto & (i(X)\tau)\tau + i(X)\d\eta + \left(i(X)\eta\right)\eta
	\end{array}
$$

\begin{prop}\label{prop:Reeb-vector-fields}
	Given a cocontact manifold $(M,\tau,\eta)$, there exist two vector fields $R_t,R_s\in\X(M)$, called \textbf{Reeb vector fields}, satisfying the conditions
	$$ \begin{dcases}
		i(R_t)\tau = 1\,,\\
		i(R_t)\eta = 0\,,\\
		i(R_t)\d\eta = 0\,,
	\end{dcases}\qquad\begin{dcases}
		i(R_s)\tau = 0\,,\\
		i(R_s)\eta = 1\,,\\
		i(R_s)\d\eta = 0\,.
	\end{dcases} $$
	$R_t$ is the \textbf{time Reeb vector field} and $R_s$ is the \textbf{contact Reeb vector field}.
\end{prop}

Taking into account the isomorphism $\flat$ introduced above, we can give an alternative definition of the Reeb vector fields:
$$ R_t = \flat^{-1}(\tau)\ ,\quad R_s = \flat^{-1}(\eta)\,. $$

\begin{thm}[Darboux theorem for cocontact manifolds]\label{thm:Darboux-cocontact}
	Let $(M,\tau,\eta)$ be a cocontact manifold. Then, for every point $p\in M$, there exists a local chart $(U; t, q^i, p_i, s)$ around $p$ such that
	$$ \restr{\tau}{U} = \d t\ ,\quad \restr{\eta}{U} = \d s - p_i\d q^i\,. $$
	These coordinates are called \textbf{canonical} or \textbf{Darboux} coordinates. Moreover, in Darboux coordinates, the Reeb vector fields are
	$$ \restr{R_t}{U} = \parder{}{t}\ ,\quad \restr{R_s}{U} = \parder{}{s}\,. $$
\end{thm}
\begin{proof}
	The idea of the proof is the following: in every cocontact manifold there is a contact foliation made of leaves of the distribution generated by $\ker\tau$ (which are $(2n+1)$-submanifolds) and, on each of them, the restriction of $\eta$ is a contact form. Now, we can take local coordinates adapted to the foliation and, on each leaf, use the Darboux theorem for contact manifolds. Hence, we have local coordinates $(\tilde t, q^i, p_i, s)$ such that $\eta = \d s - p_i\d q^i$. Finally, we can write $\tau$ as a combination of all these coordinates, $\tau = f(\tilde t, q^i, p_i, s)\d \tilde t$, where $f$ is a non-vanishing function, and then we can redefine the coordinate $t$.
\end{proof}

\begin{rmrk}\rm
    Consider a cocontact manifold $(M,\tau,\eta)$ such that $\tau = \d t$. Then, the two-form
    $$ \Omega = e^{-t}(\d\eta + \eta\wedge\tau) $$
    is a symplectic form on $M$. Moreover, the two-form $\Omega$ is exact, since
    $$ \Omega = \d(e^{-t}\eta)\,. $$
\end{rmrk}

\subsection{Cocontact manifolds as Jacobi manifolds}

It is known that symplectic, cosymplectic and contact manifolds are particular instances of a more general geometric structure: Jacobi manifolds. Then, it is reasonable to think that cocontact manifolds should also be Jacobi manifolds.

\subsubsection*{Jacobi manifolds}

\begin{dfn}
    A \textbf{Jacobi manifold} is a triple $(M,\Lambda,E)$, where $\Lambda$ is a bivector field on $M$, i.e. a skew-symmetric 2-contravariant tensor field, and $E$ is a vector field on $M$, such that
    $$ [\Lambda,\Lambda] = 2E\wedge\Lambda\ ,\qquad \Lie_E\Lambda = [E,\Lambda] = 0\,, $$
    where $[\cdot,\cdot]$ denotes the Schouten--Nijenhuis bracket {\rm\cite{Nij1955, Sch1953}}.
\end{dfn}

Given a Jacobi manifold $(M,\Lambda,E)$, we can define a bilinear map on the space of smooth functions on $M$, called the \textbf{Jacobi bracket}, given by
$$ \{f,g\} = \Lambda(\d f,\d g) + fE(g) - gE(f)\,. $$
The Jacobi bracket is bilinear, skew-symmetric, satisfies the Jacobi identity and fulfills the weak Leibniz rule
\begin{equation}\label{eq:weak-Leibniz-rule}
    \supp (\{f,g\}) \subseteq \supp (f) \cap \supp (g)\,.
\end{equation}
Hence, the set of smooth maps $\Cinfty(M)$ equipped with the Jacobi bracket $\{\cdot,\cdot\}$ is a local Lie algebra in the sense of Kirillov.

The converse is also true: given a local Lie algebra structure on $\Cinfty(M)$, one can define a Jacobi structure $(\Lambda,E)$ on $M$ such that its Jacobi bracket coincides with the bracket of the local Lie algebra bracket \cite{Kir1976,Lic1978}.

\begin{dfn}
    Given a Jacobi manifold $(M,\Lambda,E)$, we can define the morphism of vector bundles
    $$ \hat\Lambda\colon\T^\ast M\to\T M\,,$$
    given by $\hat\Lambda(\alpha) = \Lambda(\alpha,\cdot)$. This morphism of vector bundles induces a morphism of $\Cinfty(M)$-modules $\hat\Lambda\colon\Omega^1(M)\to\X(M)$.
\end{dfn}

Given a smooth function $f\in\Cinfty(M)$, we define the \textbf{Hamiltonian vector field} associated with $f$ as
$$ X_f = \hat\Lambda(\d f) + fE\,. $$
The \textbf{characteristic distribution} $\C$ of $(M,\Lambda,E)$ is spanned by these vector fields $X_f$. It is integrable, and it can be defined in terms of the bivector $\Lambda$ and the vector field $E$ as
$$ \C = \hat\Lambda(\cT M) + \langle E \rangle\,. $$

In general, the morphism $\hat\Lambda$ is not an isomorphism, as can be seen in the following example.

\begin{exmpl}[Contact manifolds]\rm
    Consider a contact manifold $(M,\eta)$. We have the $\Cinfty(M)$-module isomorphism $\flat\colon\X(M)\to\Omega^1(M)$ given by
    $$ \flat(X) = i(X)\d\eta + (i(X)\eta)\eta\,, $$
    and we denote its inverse by $\sharp = \flat^{-1}$. The Reeb vector field is $R = \sharp\eta$.
    
    Then we can define a Jacobi structure $(\Lambda,E)$ on $M$ given by
    $$ \Lambda(\alpha,\beta) = -\d\eta(\sharp\alpha,\sharp\beta)\ ,\quad E = -R\,. $$
    In this case, the morphism $\hat\Lambda$ is given by
    $$ \hat\Lambda(\alpha) = \sharp(\alpha) - \alpha(R)R\,, $$
    and it is clear that $\hat\Lambda$ is not an isomorphism, since $\ker\hat\Lambda = \langle\eta\rangle$ and $\Ima\hat\Lambda = \ker\eta$.
\end{exmpl}

\begin{exmpl}[Poisson manifolds]\rm
    A \textbf{Poisson manifold} is a smooth manifold $M$ such that $\Cinfty(M)$ has a Lie bracket satisfying the Leibniz rule
    $$ \{fg,h\} = f\{g,h\} + \{f,h\}g\,. $$
    It can be seen that this implies the weak Leibniz rule \eqref{eq:weak-Leibniz-rule}, thus giving a local Lie algebra structure to the set of smooth functions on $M$. One can prove that the Jacobi bracket of a Jacobi manifold $(M,\Lambda,E)$ is Poisson if and only if $E=0$. Taking this into account, a Poisson manifold is a Jacobi manifold with $E=0$. We can redefine the notion of Poisson manifold as a couple $(M,\Lambda)$, where $\Lambda$ is a bivector such that $[\Lambda,\Lambda] = 0$.

    We know that both symplectic and cosymplectic manifolds are Poisson. Hence, they are also Jacobi manifolds with
    $$ \Lambda(\alpha,\beta) = \omega(\sharp\alpha,\sharp\beta)\ ,\quad E = 0\,, $$
    where $\omega$ is the 2-form of the symplectic or cosymplectic manifold and $\sharp$ is the inverse of the flat morphism $\flat$ of the symplectic or cosymplectic manifold.
\end{exmpl}

\subsubsection*{Jacobi structure of cocontact manifolds}

Consider now a cocontact manifold $(M,\tau,\eta)$. In Darboux coordinates, the flat isomorphism reads
\begin{align*}
    \flat\left(\parder{}{t}\right) &= \d t = \tau\,,\\
    \flat\left(\parder{}{q^i}\right) &= \d p_i - p_i\d s + p_i^2\d q^i\,,\\
    \flat\left(\parder{}{p_i}\right) &= -\d q^i\,,\\
    \flat\left(\parder{}{s}\right) &= \d s - p_i\d q^i = \eta\,.
\end{align*}
Its inverse is the isomorphism $\sharp\colon\Omega^1(M)\to\X(M)$, which is given by
\begin{align*}
    \sharp\d t &= \parder{}{t}\,,\\
    \sharp\d q^i &= -\parder{}{p_i}\,,\\
    \sharp\d p_i &= p_i\parder{}{s} + \parder{}{q^i}\,,\\
    \sharp\d s &= \parder{}{s} - p_i\parder{}{p_i}\,.
\end{align*}

We can define a 2-contravariant skew-symmetric tensor $\Lambda$:
$$ \Lambda(\alpha,\beta) = -\d\eta(\sharp\alpha,\sharp\beta)\,, $$
which in Darboux coordinates reads
$$ \Lambda = \parder{}{q^i}\wedge\parder{}{p_i} - p_i\parder{}{p_i}\wedge\parder{}{s}\,. $$

\begin{prop}
    Let $(M,\tau,\eta)$ be a cocontact manifold. Then, $(\Lambda,E)$, where $E = -R_s$, is a Jacobi structure on $M$.
\end{prop}

The previous proposition can be proved using Darboux coordinates. Using Darboux coordinates, the Jacobi bracket reads
$$ \{f,g\} = \parder{f}{q^i}\parder{g}{p_i} - \parder{g}{q^i}\parder{f}{p_i} - p_i\left( \parder{f}{p_i}\parder{g}{s} - \parder{g}{p_i}\parder{f}{s} \right) - f\parder{g}{s} + g\parder{f}{s}\,. $$
In particular, one has
$$ \{q^i,q^j\} = \{p_i, p_j\} = 0\,,\qquad \{q^i, p_j\} = \delta_j^i\,,\qquad \{q^i, s\} = -q^i\,,\qquad \{p_i,s\} = -2p_i\,. $$

\subsubsection*{Legendrian submanifolds}

In the previous section we have seen that every cocontact manifold $(M,\tau,\eta)$ is a Jacobi manifold taking $\Lambda(\alpha,\beta) = -\d\eta(\sharp\alpha,\sharp\beta)$ and $E = -R_s$. We will study a special kind of submanifolds associated to this structure.

In first place, we find a new expression for the morphism $\hat\Lambda$ which is easier to manipulate. We have
$$ \alpha = \flat\sharp\alpha = i_{\sharp\alpha}\d\eta + (i_{\sharp\alpha}\eta)\eta + (i_{\sharp\alpha}\tau)\tau\,, $$
and contracting with the Reeb vector fields $R_s,R_t$, we get
$$ \alpha(R_t) = i_{\sharp\alpha}\tau\,,\qquad \alpha(R_s) = i_{\sharp\alpha}\eta\,. $$
Now,
\begin{align*}
    \beta(\hat\Lambda(\alpha)) &= -i_{\sharp\beta}i_{\sharp\alpha}\d\eta\\
    &= i_{\sharp\alpha}i_{\sharp\beta}\d\eta\\
    &= i_{\sharp\alpha}\left( \beta - (i_{\sharp\beta}\eta)\eta - (i_{\sharp\beta}\tau)\tau \right)\\
    &= \beta(\sharp\alpha) - \beta(R_s)\alpha(R_s) - \beta(R_t)\alpha(R_t)\\
    &= \beta\left( \sharp\alpha - \alpha(R_s)R_s - \alpha(R_t)R_t \right)\,,
\end{align*}
and thus
$$ \hat\Lambda(\alpha) = \sharp\alpha - \alpha(R_s)R_s - \alpha(R_t)R_t\,. $$

It is clear that $\tau,\eta\in\ker\hat\Lambda$. Consider $\beta\in\ker\hat\Lambda$. Then, $\sharp\beta = \beta(R_s)R_s + \beta(R_t)R_t$. Using the morphism $\flat$, we have that
$$ \beta = \beta(R_s)\eta + \beta(R_t)\tau\,. $$
Hence, we have that $\ker\hat\Lambda = \langle\tau,\eta\rangle$.

\begin{dfn}
    Let $(M,\tau,\eta)$ be a cocontact manifold and consider a submanifold $N\subset M$. We define the \textbf{cocontact orthogonal complement} of $N$ as
    $$ \T_pN^\perp = \hat\Lambda\left((T_pN)^\circ\right)\,. $$
\end{dfn}

Notice that $v\in\T_pN^\perp$ if, and only if, there exists $\alpha\in\cT_pN$ such that $\hat\Lambda(\alpha) = v$ and $\alpha(u) = 0$ for every $u\in\T_pN$.

\begin{dfn}
    Let $(M,\tau,\eta)$ be a cocontact manifold and consider a submanifold $N\hookrightarrow M$.
    \begin{enumerate}[{\rm (i)}]
        \item $N$ is said to be \textbf{coisotropic} if $\T N^\perp\subset\T N$.
        \item $N$ is said to be \textbf{isotropic} if $\T N \subset\T N^\perp$.
        \item $N$ is said to be \textbf{Legendrian} if $\T N = \T N^\perp$.
    \end{enumerate}
\end{dfn}

The following theorem characterizes isotropic and Legendrian submanifolds.

\begin{thm}
    Consider a cocontact manifold $(M,\tau,\eta)$ with $\dim M = 2n+2$ and a submanifold $j\colon N\hookrightarrow M$. Then,
    \begin{enumerate}[{\rm (i)}]
        \item $N$ is isotropic if and only if $j^\ast\eta = 0$ and $j^\ast\tau = 0$.
        \item $N$ is Legendrian if and only if $j^\ast\eta = 0$, $j^\ast\tau = 0$ and $\dim N = n$.
    \end{enumerate}
\end{thm}
\begin{proof}
    \begin{enumerate}[(i)]
        \item Suppose $\T N \subset \T N^\perp$. Let $v\in\T_p N$, then $v\in\T_pN^\perp$ and hence there exists some $\alpha\in\cT_pN$ such that $\hat\Lambda(\alpha) = v$ and $\alpha(u) = 0$ for every $u\in\T_pN$. Then,
        $$ \eta(v) = \eta(\hat\Lambda(\alpha)) = \eta(\sharp\alpha) - \alpha(R_s)R_s - 0 = 0\,. $$
        Since $\eta_p(v) = 0$ for every $p\in N$ and every $v\in\T_pN$, we have $j^\ast\eta = 0$. Analogously, $j^\ast\tau = 0$.
        
        Conversely, suppose $j^\ast\eta = 0$ and $j^\ast\tau = 0$. Let $v\in\T_pN$. We want to check that $v\in\T_pN^\perp$. Consider
        $$ \alpha = i_v\d\eta\,. $$
        Since $j^\ast\eta = 0$, we have $j^\ast\d\eta = \d j^\ast\eta = 0$ and thus $i_ui_v\d\eta = \alpha(u) = 0$ for every $u\in\T_pN$. Now, we want to check that
        $$ v = \hat\Lambda(\alpha) = \sharp\alpha - \alpha(R_s)R_s - \alpha(R_t)R_t = \sharp\alpha\,. $$
        Since $\flat$ is an isomorphism, this is the same as checking that $\flat(v) = \alpha$, which is obvious taking into account that $j^\ast\eta = 0$ and $j^\ast\tau = 0$. Finally, we have already proved that $\alpha(\T_pN) = 0$.
        
        \item Suppose that $\T N = \T N^\perp$ and let $k = \dim N$. Thus, at any point $p\subset N$,
        $$ \dim\T_pN^\circ = 2n+2 - \dim\T_pN = 2n+2-k\,. $$
        Since $\T N\subset\T N^\perp$, we have that $\tau_p,\eta_p\in\T_pN^\circ$ and then $\ker\hat\Lambda_p\subset\T_pN^\circ$. Now, since $T_pN = T_pN^\perp = \hat\Lambda_p(\T_pN^\circ)$,
        $$ \underbrace{\dim\T_pN^\circ}_{2n+2-k} = \underbrace{\dim\hat\Lambda_p(\T_pN^\circ)}_{k} + \underbrace{\dim\ker\hat\Lambda_p}_{2}\,, $$
        which implies that $k = n$.
        
        Conversely, we have that $2n+2-n = \dim\T_pN^\perp + 2$ and thus $\dim\T_pN^\perp = n = \dim\T_pN$. Since $\T_pN\subset\T_pN^\perp$, it is clear that $\T_pN = \T_pN^\perp$.
    \end{enumerate}
\end{proof}

\section{Cocontact Hamiltonian systems}

We can use the cocontact geometric framework developed in the previous section to deal with Hamiltonian systems with dissipation and with explicit time dependence.

\begin{dfn}
	A \textbf{cocontact Hamiltonian system} is a tuple $(M,\tau,\eta,H)$ where $(M,\tau,\eta)$ is a cocontact manifold and $H\in \Cinfty (M)$ is a Hamiltonian function.
	
	The \textbf{cocontact Hamilton equations} for a curve $\psi\colon I\subset \R\to M$ are
	\begin{equation}\label{eq:Ham-eq-cocontact-sections}
		\begin{dcases}
			i(\psi')\d \eta = \d H-(\Lie_{R_s}H)\eta-(\Lie_{R_t}H)\tau\,,
			\\
			i(\psi')\eta = -H\,,
			\\
			i(\psi')\tau = 1\,,	 
		\end{dcases}
	\end{equation}
	where $\psi'\colon I\subset\R\to\T M$ is the canonical lift of the curve $\psi$ to the tangent bundle $\T M$.
\end{dfn}

Let us express these equations using Darboux coordinates (see Theorem \ref{thm:Darboux-cocontact}). Consider the curve $\psi(r)=(f(r),q^i(r),p_i(r),s(r))$. The third equation in \eqref{eq:Ham-eq-cocontact-sections} imposes that $f(r)=r$, thus we will denote $r\equiv t$. The other equations become:

\begin{equation*}
   \begin{dcases}
		\dot q^i =\frac{\partial H}{\partial p_i}\,,
		\\
	   \dot p_i =-\left(\frac{\partial H}{\partial q^i}+p_i\frac{\partial H}{\partial s}\right)\,,
		\\
	   \dot s =p_ i\frac{\partial H}{\partial p_i}-H\,.		
	\end{dcases}	  
\end{equation*}

We can give another interpretation of these equations using vector fields.

\begin{dfn}
	Let $(M,\tau,\eta,H)$ be a cocontact Hamiltonian system. The \textbf{cocontact Hamiltonian equations} for a vector field $X\in\X(M)$ are:
	\begin{equation}\label{eq:Ham-eq-cocontact-vectorfields}
		\begin{dcases}
			   i(X)\d \eta = \d H-(\Lie_{R_s}H)\eta-(\Lie_{R_t}H)\tau\,,
			\\
			i(X)\eta = -H\,,
			\\
			i(X)\tau = 1\,.  
		\end{dcases}	   
	\end{equation}
\end{dfn}
Equations \eqref{eq:Ham-eq-cocontact-vectorfields} can be rewritten using the isomorphism $\flat$ as
$$
    \flat(X)=\d H-\left(\Lie_{R_s}H+H\right)\eta+\left(1-\Lie_{R_t}H\right)\tau\,.
$$
Therefore, it is clear that the cocontact Hamilton equations
\eqref{eq:Ham-eq-cocontact-vectorfields} have a unique solution. The solution to these equations is called the \textbf{cocontact Hamiltonian vector field}, and will be denoted by $X_H$. Its local expression is
$$ X_H = \parder{}{t} + \parder{H}{p_i}\parder{}{q^i} - \left(\parder{H}{q^i} + p_i\parder{H}{s}\right)\parder{}{p_i} + \left(p_i\parder{H}{p_i} - H\right)\parder{}{s}\,. $$

\begin{rmrk}\rm
    It can be seen that the vector field $X_H - R_t$ coincides with the one for the associated Jacobi manifold given in the preceding section.
\end{rmrk}

\begin{prop}
	Let $X$ be a vector field in $M$. Then every integral curve $\psi:U\subset \mathbb{R}\rightarrow M$ of $X$ satisfies the cocontact equations \eqref{eq:Ham-eq-cocontact-sections} if, and only if, $X$ is a solution to \eqref{eq:Ham-eq-cocontact-vectorfields}.
\end{prop}
\begin{proof}
	This is a direct consequence of equations \eqref{eq:Ham-eq-cocontact-sections} and \eqref{eq:Ham-eq-cocontact-vectorfields}, and the fact that any point of $M$ is in the image of an integral curve of $X$.
\end{proof}

\section{Cocontact Lagrangian systems}\label{sec:Lagrangian-formalism}
\subsection{Lagrangian phase space and geometric structures}\label{subsub:Lagrangian-geometric-structures}

Let $Q$ be a smooth $n$-dimensional manifold. Consider the product manifold $\R\times\T Q\times \R$ endowed with natural coordinates $(t, q^i, v^i, s)$. We have the canonical projections
\begin{align*}
	\tau_1&\colon \R\times\T Q\times\R\to\R\ ,&& \tau_1(t, v_q, s) = t\,,\\
	\tau_2&\colon\R\times\T Q\times\R\to\T Q\ ,&& \tau_2(t, v_q, s) = v_q\,,\\
	\tau_3&\colon\R\times\T Q\times\R\to\R\ ,&& \tau_3(t, v_q, s) = s\,,\\
	\tau_0&\colon\R\times\T Q\times\R\to \R\times Q\times\R\ ,&& \tau_0(t, v_q, s) = (t, q, s)\,. 
\end{align*}

Notice that $\tau_2$ and $\tau_0$ are the projection maps of two vector bundle structures. Usually, we will have in mind the projection $\tau_0$. In fact, the vector bundle $\R\times\T Q\times\R\overset{\tau_0}{\longrightarrow}\R\times Q\times\R$ is the pull-back of the tangent bundle with respect to the map $\R\times Q\times\R\to Q$. 

Our goal is to develop a time-dependent contact Lagrangian formalism on the manifold $\R\times\T Q\times\R$. The usual geometric structures of the tangent bundle can be naturally extended to the cocontact Lagrangian phase space $\R\times\T Q\times\R$. In particular, the vertical endomorphism of $\T(\T Q)$ yields a \textbf{$\tau_0$-vertical endomorphism} $\mathcal{J}\colon\T(\R\times\T Q\times\R)\to\T(\R\times\T Q\times\R)$. In the same way, the Liouville vector field on the fibre bundle $\T Q$ gives a \textbf{Liouville vector field} $\Delta\in\X(\R\times\T Q\times\R)$. Actually, this Liouville vector field coincides with the Liouville vector field of the vector bundle $\tau_0$. The local expressions of these objects in Darboux coordinates are
$$ {\cal J} = \frac{\partial}{\partial v^i} \otimes \d q^i \,,\quad \Delta =  v^i\, \frac{\partial}{\partial v^i} \,. $$

\begin{dfn}\label{dfn:holonomic-path}
    Let ${\bf c} \colon\R \rightarrow \R\times Q\times\R$ be a path, with ${\bf c} = (\mathbf{c}_1,\mathbf{c}_2,\mathbf{c}_3)$. The \textbf{prolongation} of ${\bf c}$ to $\R\times\T Q\times\R$ is the path
    $$ {\bf \tilde c} =  (\mathbf{c}_1,\mathbf{c}_2',\mathbf{c}_3) \colon \R \longrightarrow \R\times\T Q \times \R  \,, $$
    where $\mathbf{c}_2'$ is the velocity curve of~$\mathbf{c}_2$. Every path ${\bf \tilde c}$ which is the prolongation of a path ${\bf c} \colon\R \rightarrow \R\times Q\times\R$ is said to be \textbf{holonomic}. A vector field 
    $\Gamma \in \X(\R\times\T Q \times \R)$ is said to satisfy the \textbf{second-order condition} (for short: it is a {\sc sode}) when all of its integral curves are holonomic.
\end{dfn}

This definition can be equivalently expressed in terms of the canonical structures defined above:

\begin{prop}
    A vector field $\Gamma  \in \X(\R\times\T Q\times\R)$ is a {\sc sode} if, and only if, ${\cal J} \circ \Gamma = \Delta$.
\end{prop}

In natural coordinates, if ${\bf c}(r)=(t(r), c^i(r), s(r))$, then
$$ {\bf \tilde c}(r) = \left(t(r), c^i(r),\frac{\d c^i}{\d r}(r), s(r) \right) \,. $$
The local expression of a {\sc sode} is
\begin{equation}\label{eq:sode-local-expression}
    \Gamma = f\parder{}{t} + v^i \parder{}{q^i} + G^i \parder{}{v^i} + g\,\frac{\partial}{\partial s}\,.
\end{equation}
So, in coordinates a {\sc sode} defines a system of
differential equations of the form
$$
    \frac{\d t}{\d r} = f(t, q, \dot q, s)\,,\quad \frac{\d^2 q^i}{\d r^2} = G^i(t, q,\dot q,s) \,, \quad \frac{\d s}{\d r} = g(t, q,\dot q,s)  \:.
$$

\subsection{Cocontact Lagrangian systems}\label{sec:cocontact-lagrangian-systems}

\begin{dfn}\label{dfn:lagrangian}
    A \textbf{Lagrangian function} is a function $\L \colon\R\times\T Q\times\R\to\R$. The \textbf{Lagrangian energy} associated to $\L$ is the function $E_\L := \Delta(\L)-\L\in\Cinfty(\R\times\T Q\times\R)$. The \textbf{Cartan forms} associated to $\L$ are defined as
	\begin{equation}\label{eq:thetaL}
		\theta_\L = \transp{\cal J}\circ\d\L\in \Omega^1(\R\times\T Q\times\R)\,,\quad\omega_\L = -\d\theta_\L\in \Omega^2(\R\times\T Q\times\R) \,.
	\end{equation}
	The \textbf{contact Lagrangian form} is
	$$ \eta_\L=\d s-\theta_\L\in\Omega^1(\R\times\T Q\times\R)\,. $$
	Notice that $\d\eta_\L=\omega_\L$. The couple $(\R\times\T Q\times\R,\L)$ is a \textbf{cocontact Lagrangian system}.
\end{dfn}

Taking natural coordinates $(t, q^i, v^i, s)$ in $\R\times\T Q\times\R$, the form $\eta_\L$ is written as
\begin{equation}\label{eq:etaL}
	\eta_\L = \d s - \frac{\partial\L}{\partial v^i} \,\d q^i \:,
\end{equation}
and consequently
\begin{equation*}
	\d\eta_\L = -\parderr{\L}{t}{v^i}\d t\wedge\d q^i -\parderr{\L}{q^j}{v^i}\d q^j\wedge\d q^i -\parderr{\L}{v^j}{v^i}\d v^j\wedge\d q^i -\parderr{\L}{s}{v^i}\d s\wedge\d q^i \ .
\end{equation*}

Notice that, in general, given a cocontact Lagrangian system $(\R\times\T Q\times\R,\L)$, the family $(\R\times\T Q\times\R,\tau=\d t,\eta_\L, E_\L)$ is not a cocontact Hamiltonian system because condition $\tau\wedge\eta\wedge(\d\eta_\L)^n\neq 0$ is not fulfilled. The Legendre map characterizes which Lagrangian functions will give cocontact Hamiltonian systems.

\begin{dfn}
    Given a Lagrangian $\L \colon\R\times \T Q\times\R \to \R$, its \textbf{Legendre map} is the fibre derivative of~$\L$, considered as a function on the vector bundle $\tau_0 \colon\R\times \T Q\times\R \to \R\times Q \times \R$; that is, the map $\F\L \colon \R\times\T Q \times\R \to \R\times\cT Q \times \R$ given by
    $$
        \F\L (t,v_q,s) = \left( t,\F\L(t,\cdot,s) (v_q),s \right)\,,
    $$
    where $\F\L(t,\cdot,s)$ is the usual Legendre map associated to the Lagrangian $\L(t,\cdot,s)\colon\T Q\to\R$ with $t$ and $s$ freezed.
\end{dfn}

Notice that the Cartan forms can also be defined as
$$
    \theta_\L={\cal FL}^{\;*}(\pi^*\theta_0)\,,\quad\omega_\L={\cal FL}^{\;*}(\pi^*\omega_0)\,,
$$
where $\theta_0$ and $\omega_0 = -\d\theta_0$ are the canonical 1- and 2-forms of the cotangent bundle and $\pi$ is the natural projection $\pi\colon\R\times\cT Q\times\R\to\cT Q$ (see Example \ref{ex:canonical-cocontact-manifold}).

\begin{prop}\label{prop:regular-Lagrangian}
    For a Lagrangian function $\L$ the following conditions are equivalent:
    \begin{enumerate}[{\rm(i)}]
        \item The Legendre map $\F \L$ is a local diffeomorphism.
        \item The fibre Hessian $\F^2\L \colon\R\times\T Q\times\R \longrightarrow (\R\times\cT Q\times\R)\otimes (\R\times\cT Q\times\R)$ of $\L$ is everywhere nondegenerate (the tensor product is of vector bundles over $\R\times Q \times \R$).
        \item $(\R\times\T Q\times\R,\d t, \eta_\L)$ is a cocontact manifold.
    \end{enumerate}
\end{prop}

The proof of this result can be easily done using natural coordinates, where
$$
    \F\L:(t,q^i,v^i, s) \longrightarrow  \left(t,q^i,\parder{L}{v^i}, s\right)\,,
$$
and
$$
    \F^2 \L(t,q^i,v^i,s) = (t,q^i,W_{ij},s)\,,\quad\mbox{with }\ W_{ij} = \left( \parderr{\L}{v^i}{v^j}\right)\,.
$$
And the three conditions in the above proposition are equivalent to say that the matrix $W= (W_{ij})$ is everywhere regular, hence they are equivalent.


\begin{dfn}
    A Lagrangian function $\L$ is said to be \textbf{regular} if the equivalent conditions in Proposition \ref{prop:regular-Lagrangian} hold. Otherwise $\L$ is called a \textbf{singular} Lagrangian. In particular, $\L$ is said to be \textbf{hyperregular} if $\F\L$ is a global diffeomorphism.
\end{dfn}

\begin{rmrk}\rm
    As a result of the preceding definitions and results, every \emph{regular} cocontact Lagrangian system has associated the cocontact Hamiltonian system $(\R\times\T Q\times\R,\d t, \eta_\L, E_\L)$.
\end{rmrk}

Given a regular cocontact Lagrangian system $(\R\times\T Q\times\R,\L)$,
from Proposition \ref{prop:Reeb-vector-fields} we have that
the \textbf{Reeb vector fields} 
$R_t^\L,R_s^\L\in\X(\R\times\T Q\times\R)$ 
for this system are uniquely determined by the relations
$$ \begin{dcases}
		i(R_t^\L)\d t = 1\,,\\
		i(R_t^\L)\eta_\L = 0\,,\\
		i(R_t^\L)\d\eta_\L = 0\,,
	\end{dcases}\qquad\begin{dcases}
		i(R_s^\L)\d t = 0\,,\\
		i(R_s^\L)\eta_\L = 1\,,\\
		i(R_s^\L)\d\eta_\L = 0\,.
	\end{dcases} $$
Their local expressions are
\begin{align*}
    R_t^\L &= \parder{}{t} - W^{ij}\parderr{\L}{t}{v^j}\parder{}{v^i}\,,\\
    R_s^\L &= \parder{}{s} - W^{ij}\parderr{\L}{s}{v^j}\parder{}{v^i}\,,
\end{align*}
where $W^{ij}$ is the inverse of the Hessian matrix of the Lagrangian $\L$, namely $W^{ij}W_{jk}=\delta^i_k$.

Notice that if the Lagrangian function $\L$ is singular, the Reeb vector fields are not uniquely determined. This will be discussed in Section \ref{sec:singular}.

\subsection{The Herglotz--Euler--Lagrange equations}

\begin{dfn}\label{dfn:Euler-Lagrange-equations}
	Let $(\R\times\T Q\times\R,\L)$ be a regular contact Lagrangian system.

	The \textbf{Herglotz--Euler--Lagrange equations} for a holonomic curve ${\bf\tilde c}\colon I\subset\R \to\R\times\T Q\times\R$ are
	\begin{equation}\label{eq:Euler-Lagrange-curve}
		\begin{dcases}
			i({\bf\tilde c}')\d\eta_\L = \left(\d E_\L - (\Lie_{R_t^\L}E_\L)\d t - (\Lie_{R_s^\L}E_\L)\eta_\L\right)\circ{\bf\tilde c}\,,\\
			i({\bf\tilde c}')\eta_\L = - E_\L\circ{\bf\tilde c}\,,\\
			i({\bf\tilde c}')\d t = 1\,,
		\end{dcases}
	\end{equation}
	where ${\bf\tilde c}'\colon I\subset\R\to\T(\R\times\T Q\times\R)$ denotes the canonical lifting of ${\bf\tilde c}$ to $\T(\R\times\T Q\times\R)$.
	
	The \textbf{cocontact Lagrangian equations} for a vector field $X_\L\in\X(\R\times\T Q\times\R)$ are 
	\begin{equation}\label{eq:Euler-Lagrange-field}
		\begin{cases}
			i(X_\L)\d\eta_\L = \d E_\L-(\Lie_{R_t^\L}E_\L)\d t-(\Lie_{R_s^\L}E_\L)\eta_\L\,,\\
			i(X_\L)\eta_\L = -E_\L \,,\\
			i(X_\L)\d t = 1\,.
		\end{cases}
	\end{equation}
	The vector field which is the only solution to these equations is called the \textbf{cocontact Lagrangian vector field}.
\end{dfn}

 Notice that a cocontact Lagrangian vector field is a cocontact Hamiltonian vector field for the function $E_\L$ (and the cocontact structure $(\d t,\eta_\L)$).

In natural coordinates, for a holonomic curve
${\bf\tilde c}(r)=(t(r), q^i(r),\dot q^i(r),s(r))$,
equations \eqref{eq:Euler-Lagrange-curve} are
\begin{align}
	\dot t &= 1\,,\label{eq:EL-curve-coords-1}\\
	\dot s &= \L\,,\label{eq:EL-curve-coords-2}\\
	\dot t\parderr{\L}{t}{v^i} + \dot q^j\parderr{\L}{q^j}{v^i} + \ddot q^j\parderr{\L}{v^j}{v^i} + \dot s\parderr{\L}{s}{v^i} - \parder{\L}{q^i} = \frac{\d}{\d r}\left(\parder{\L}{v^i}\right) - \parder{\L}{q^i} &= \parder{\L}{s}\parder{\L}{v^i}\,.\label{eq:EL-curve-coords-3}
\end{align}
Condition \eqref{eq:EL-curve-coords-1} implies that $t(r) = r + k$. This justifies the usual identification $t = r$. Meanwhile, for a vector field $X_\L = f\dparder{}{t} + F^i\dparder{}{q^i} + G^i\dparder{}{v^i} + g\dparder{}{s}$, equations \eqref{eq:Euler-Lagrange-field} read
\begin{align}
	(F^j - v^j)\parderr{\L}{t}{v^j} &= 0\,,\label{eq:EL-field-coords-1}\\
	f\parderr{\L}{t}{v^i} + F^j\parderr{\L}{q^j}{v^i} + G^j\parderr{\L}{v^j}{v^i} + g\parderr{\L}{s}{v^i} - \parder{\L}{q^i} - (F^j - v^j)\parderr{\L}{q^i}{v^j} &= \parder{\L}{s}\parder{\L}{v^i}\,,\label{eq:EL-field-coords-2}\\
	(F^j - v^j)\parderr{\L}{v^i}{v^j} &= 0\,,\label{eq:EL-field-coords-3}\\
	(F^j - v^j)\parderr{\L}{s}{v^j} &= 0\,,\label{eq:EL-field-coords-4}\\
	\L + \parder{\L}{v^j}(F^j-v^j) - g &= 0\,,\label{eq:EL-field-coords-5}\\
	f &= 1\label{eq:EL-field-coords-6}\,,
\end{align}
where we have used the relations
\begin{equation*}
    \Lie_{R_t^\L} E_\L = - \parder{\L}{t}\ ,\quad\Lie_{R_s^\L} E_\L = - \parder{\L}{s}\,,
\end{equation*}
which can be easily proved taking coordinates.

\begin{thm}\label{thm:regular-lagrangian}
    If $\L$ is a regular Lagrangian, then $X_\L$ is a {\sc sode} and equations \eqref{eq:EL-field-coords-1}--\eqref{eq:EL-field-coords-6} become
    \begin{align}
        f &= 1\,,\label{eq:regular-EL-field-coords-1}\\
        g &= \L\,,\label{eq:regular-EL-field-coords-2}\\
        \parderr{\L}{t}{v^i} + v^j\parderr{\L}{q^j}{v^i} + G^j\parderr{\L}{v^j}{v^i} + \L\parderr{\L}{s}{v^i} - \parder{\L}{q^i}
        &= \parder{\L}{s}\parder{\L}{v^i}\,,\label{eq:regular-EL-field-coords-3}
    \end{align}
    which, for the integral curves of $X_\L$, are the \textbf{Herglotz--Euler--Lagrange equations} \eqref{eq:EL-curve-coords-1}, \eqref{eq:EL-curve-coords-2} and \eqref{eq:EL-curve-coords-3}.

    This {\sc sode} $X_\L\equiv\Gamma_\L$ is called the \textbf{Herglotz--Euler--Lagrange vector field} associated to the Lagrangian function $\L$.
\end{thm}
\begin{proof}
    This results follows from the coordinate expressions. If $\L$ is a regular Lagrangian, equations \eqref{eq:EL-field-coords-3} lead to $F^i=v^i$, which are the {\sc sode} condition for the vector field $X_\L$. Then, \eqref{eq:EL-field-coords-1} and \eqref{eq:EL-field-coords-4} hold identically, and \eqref{eq:EL-field-coords-2}, \eqref{eq:EL-field-coords-5} and \eqref{eq:EL-field-coords-6} give the equations \eqref{eq:regular-EL-field-coords-1}, \eqref{eq:regular-EL-field-coords-2} and \eqref{eq:regular-EL-field-coords-3} or, equivalently, for the integral curves of $X_\L$, the Herglotz--Euler--Lagrange equations \eqref{eq:EL-curve-coords-1}, \eqref{eq:EL-curve-coords-2} and \eqref{eq:EL-curve-coords-3}.
\end{proof}

Hence, the local expression of the Herlgotz--Euler--Lagrange vector field is
\begin{equation*}
    \Gamma_\L = \parder{}{t} + v^i\parder{}{q^i} + W^{ji}\left( \parder{\L}{q^j} - \parderr{\L}{t}{v^j} - v^k\parderr{\L}{q^k}{v^j} - \L\parderr{\L}{s}{v^j} + \parder{\L}{s}\parder{\L}{v^j} \right)\parder{}{v^i} + \L\parder{}{s}\,.
\end{equation*}
An integral curve of this vector field satisfies the classical Euler--Lagrange equation for dissipative systems, as can be easily verified:
$$ \frac{\d}{\d t}\left(\parder{\L}{v^i}\right) - \parder{\L}{q^i} =\parder{\L}{s}\parder{\L}{v^i}\,,\qquad
\dot s = \L\,.
$$

\begin{rmrk}\rm
    It is interesting to point out how, in the Lagrangian formalism of dissipative systems, the expression in coordinates \eqref{eq:EL-curve-coords-2} of the second Lagrangian equation \eqref{eq:Euler-Lagrange-curve} relates the variation of the ``dissipation coordinate'' $s$ to the Lagrangian function and, from here, we can identify this coordinate with the Lagrangian action, $\displaystyle s=\int\L\,\d t$.
\end{rmrk}

\begin{rmrk}\rm Equations \eqref{eq:EL-curve-coords-2} \eqref{eq:EL-curve-coords-3} coincide with those derived from the Herglotz variational principle \cite{Her1930,Gue1996,Geo2003}.

\end{rmrk}

\section{The singular case}\label{sec:singular}

If the Hessian of the Lagrangian is not regular (in other words, the Lagrangians is singular), the construction of section \ref{sec:cocontact-lagrangian-systems} does not lead to a cocontact structure, as shown in \ref{prop:regular-Lagrangian}. In order to describe the dynamics of singular Lagrangians we require some minimal regularity conditions. Inspired by previous works on the singular case for contact systems \cite{DeLeo2019} and cosymplectic systems \cite{Chi1994} we will define precocontact structures. 

\subsection{Precocontact structures}

\begin{dfn}
    Let $M$ be a smooth manifold and consider two one-forms $\tau,\eta\in\Omega^1(M)$. We define the \textbf{characteristic distribution} of $(\tau,\eta)$ as
    $$ \C = \ker\tau\cap\ker\eta\cap\ker\d\eta\,. $$
    If $\C$ has constant rank, we say that the couple $(\tau,\eta)$ is of \textbf{class} $\cl(\tau,\eta) = \dim M - \rk \C$.
\end{dfn}

\begin{dfn}
    Consider a smooth manifold $M$. The couple $(\tau,\eta)$, where $\tau,\eta\in\Omega^1(M)$, is a \textbf{precocontact structure} on $M$ if $\d\tau = 0$, the characteristic distribution of the couple $(\tau,\eta)$ has constant rank and $\cl(\tau,\eta) = 2r+2\geq 2$. The triple $(M,\tau,\eta)$ will be called a \textbf{precocontact manifold}.
\end{dfn}

\begin{thm}[Darboux theorem]\label{thm:darbouxSing}
    Let $(M,\tau,\eta)$ be a precocontact manifold with $\dim M = m$. Then, around every point $p\in M$, there is a local chart $(U; t, q^i, p_i, s, u^j)$, where $1\leq i \leq r$, $1\leq j\leq c$, $2r+2+c=m$, such that
    $$ \restr{\tau}{U} = \d t\ ,\quad\restr{\eta}{U} = \d s - p_i\d q^i\ ,\quad \C = \left\langle \parder{}{u^j} \right\rangle\,. $$
\end{thm}
\begin{proof}
    Since the one-form $\tau$ is closed, we have that locally $\tau = \d t$. Moreover, $\ker\tau$ is an integrable distribution and thus induces a foliation $\mathcal{F}$ of $M$. Consider the leaf $\mathcal{F}_0$ corresponding to $t = t_0$. We want to see that the class of $\restr{\eta}{\mathcal{F}_0}$ is odd (the definition of class of a one-form can be found in Definition VI.1.1 in \cite{God1969}). It is clear that, on a point $p\in\mathcal{F}_0$, $\C= \C_{\restr{\eta}{\mathcal{F}_0}}$. Hence,
    $$ \cl(\restr{\eta}{\mathcal{F}_0}) = m - 1 - \rk \C_{\restr{\eta}{\mathcal{F}_0}} = m - 1 - \rk\C_{t_0} = \cl(\tau,\eta) - 1\,, $$
    and since $\cl(\tau,\eta)\geq 2$ is even, we have that $\cl(\restr{\eta}{\mathcal{F}_0})\geq 1$ is odd. Now, using Theorem VI.4.1 in \cite{God1969} we obtain the desired result.
\end{proof}



\begin{rmrk}\rm
    If $\cl(\tau,\eta)=m$ then $m$ must be even and we recover the cocontact structure introduced in Definition \ref{dfn:cocontact-manifold}.
\end{rmrk}

Given a precocontact manifold $(M,\tau,\eta)$, the morphism $\flat$ can be defined in the same way as in the regular case:
$$
	\begin{array}{rccl}
		\flat\colon & \T M & \longrightarrow & \cT M \\
		& v & \longmapsto & (i(v)\tau)\tau + i(v)\d\eta + \left(i(v)\eta\right)\eta\,.
	\end{array}
$$
However, in this case it is not an isomorphism (see Proposition \ref{prop:kernelflat}).

\begin{dfn}\label{dfn:ReebSing}
	A vector field $R_s\in\mathfrak{X}(M)$ is a \textbf{contact Reeb vector field} if
	$$
		\flat(R_s)=\eta\,.
	$$
	A vector field $R_t\in\mathfrak{X}(M)$ is a \textbf{time Reeb vector field} if
	$$
		\flat(R_t)=\tau\,.
	$$
\end{dfn}

Both kind of vector fields always exists global thanks to the theorem \ref{thm:darbouxSing}, and a global one can be constructed using partitions of the unity. Crucially, in the precocontact case the Reeb vector fields are not unique.

\begin{prop}\label{prop:kernelflat}
	Let $(M,\tau,\eta)$ be a precocontact structure on $M$. Then
	$$
		\C = \ker\tau\cap\ker\eta\cap\ker\d\eta = \ker \flat = (\Ima\flat)^\circ\,.
	$$
\end{prop}
\begin{proof}
    We will begin by proving that $\ker\tau\cap\ker\eta\cap\ker\d\eta \supseteq \ker\flat$. Let $X\in\ker\flat$, then
    \begin{equation}\label{eq:flatX}
        \flat(X) = 0\,.    
    \end{equation}
    Let $R_s,R_t$ be a contact and time Reeb vector fields respectively. Contracting both sides of equation \eqref{eq:flatX} with $R_s$, we have
    \begin{align*}
        0 &= i(R_s)i(X)\d\eta + i(R_s)((i(X)\eta)\eta) + i(R_s)((i(X)\tau)\tau)\\
        &= -i(X)(i(R_s)\d\eta)) + (i(X)\eta)(i(R_s)\eta) + (i(X)\tau)(i(R_s)\tau)\\
        &= i(X)\eta\,.
    \end{align*}
    On the other hand, contracting with $R_t$, we obtain
    \begin{align*}
        0 &= i(R_t)i(X)\d\eta + i(R_t)((i(X)\eta)\eta) + i(R_t)((i(X)\tau)\tau)\\
        &= -i(X)(i(R_t)\d\eta)) + (i(X)\eta)(i(R_t)\eta) + (i(X)\tau)(i(R_t)\tau)\\
        &= i(X)\tau\,.
    \end{align*}
    Hence, we also have that $i(X)\d\eta = 0$ and then $X\in\ker\tau\cap\ker\eta\cap\ker\d\eta$. The other inclusion is trivial.
    
    Now we will show that $\ker\flat = (\Ima\flat)^\circ$. Let $X\in\ker\flat$. By the first equality, $i(X)\tau = 0$, $i(X)\eta = 0$ and $i(X)\d\eta = 0$. Then, for every vector field $Y$,
    $$ i(X)\flat(Y) = i(X)(i(Y)\tau)\tau + i(X)i(Y)\d\eta + i(X)(i(Y)\eta)\eta = -i(Y)i(X)\d\eta = 0\,, $$
    and hence $\ker\flat\subset(\Ima\flat)^\circ$. Now, as it is clear that both subspaces of $\T_p M$ have the same dimension, we have that $\ker\flat=(\Ima\flat)^\circ$.
\end{proof}

Notice that if $R_1$ and $R_2$ are two contact (or two time) Reeb vector fields, then, $R_1-R_2\in \C$. Conversely, if $R$ is a contact (time) Reeb vector field, then $R+V$ is  a contact (or time) Reeb vector field for any $V\in\C$.

\begin{dfn}
	A \textbf{precocontact Hamiltonian system} is a family $(M,\tau,\eta,H)$ where $(M,\tau,\eta)$ is a precocontact manifold and $H\in\Cinfty(M)$ is the \textbf{Hamiltonian function} on $M$.
\end{dfn}

\begin{dfn}
	Let $(M,\tau,\eta,H)$ be a precocontact Hamiltonian system and $R_s$ and $R_t$ two contact and time Reeb vector fields respectively. The \textbf{precocontact Hamiltonian equation} for a vector field $X$ is:
	\begin{equation}\label{eq:Ham-eq-precocontact-vectorfields}
		\flat(X)=\gamma_H\,,	   
	\end{equation}
	where $\gamma_H = \d H-\left(\Lie_{R_s}H+H\right)\eta+\left(1-\Lie_{R_t}H\right)\tau$.
\end{dfn}

Equation \eqref{eq:Ham-eq-precocontact-vectorfields} has two problems in the singular case: $\gamma_H$ depends (a priori) on the Reeb vector fields we have chosen and there may not exist a solution to the equations. Both problems are solved by applying a suitable constraint algorithm.

\subsection{Constraint algorithm}

The aim of the constraint algorithm is to find a submanifold $M_f\subset M$ such that there exists a solution to the equation \eqref{eq:Ham-eq-precocontact-vectorfields} which is tangent to the submanifold $M_f$. The algorithm has two steps: the consistency condition, where we look for the submanifold where the equation has solution, and the tangency condition, where we analyse where the solution is tangent to the submanifold. These steps are applied iteratively, but the first step is specially relevant because it will solve the multiplicity of Reeb vector fields.

Consider a $\gamma_H$ with particular Reeb vector fields $R_s$ and $R_t$. $M_1$ is defined as the subset of the points of $M$ where solutions exist:
$$
	M_1=\{p\in M \mid (\gamma_H)_p\in\flat(T_pM)\}\,.
$$

\begin{thm}
    Let $\bar M_1 = \{p\in M\mid (\Lie_Y H)_p = 0\,,\ \forall Y\in\C \}$. Then $M_1 = \bar M_1$.
\end{thm}
\begin{proof}
    Let $p\in M_1$. Then $(\gamma_H)_p \in\flat(\T_pM) = (\ker\flat)_p^\circ = \C^\circ_p$. For any $Y\in\C$ we have that
    $$ 0 = (i_Y\gamma_H)_p = \big(i_Y(\d H-\left(\Lie_{R_s}H+H\right)\eta+\left(1-\Lie_{R_t}H\right)\tau)\big)_p = (\Lie_Y H)_p\,, $$
    because $i_Y\tau = i_Y\eta = 0$. Hence, $p\in \bar M_1$.
    
    Conversely, consider $p\in\bar M_1$ and $Y\in\C$. Then, $(i_Y\d H)_p = 0$, which implies $(i_Y\gamma_H)_p = 0$. Thus $(\gamma_H)_p\in(\C)^\circ_p = \flat(\T_pM)$. Since $Y$ is arbitrary, we have that $p\in M_1$.
\end{proof}

\begin{cor}\label{cor:independent-Reeb}
    Over $M_1$, the one-form $\gamma_H$ does not depend on the choice of the Reeb vector fields.
\end{cor}
\begin{proof}
    Let $\gamma_H$ as above and $\bar\gamma_H$ for $\bar R_t,\bar R_s$. Then, $\gamma_H - \bar \gamma_H = (\Lie_{(\bar R_t - R_t)}H)\tau + (\Lie_{(\bar R_s -  R_s)}H)\eta$. Since $\bar R_t - R_t, \bar R_s- R_s\in\C$, we have that $\gamma_H = \bar\gamma_H$ on $M_1$.
\end{proof}

Notice that on $M_1$, the precocontact Hamiltonian equation for the precocontact system $(M,\tau,\eta,H)$ does not depend on the choice of the Reeb vector fields $R_t,R_s$.
\bigskip



In $M_1$ (which we assume to be a submanifold) there exists a solution $X_H$ of \eqref{eq:Ham-eq-precocontact-vectorfields}, but it may not be tangent to $M_1$. Therefore, we define
$$
	M_2=\{p\in M_1 \mid (X_H)_p\in(T_pM_1)\}\,,
$$
which we also assume to be a submanifold. Iterating this procedure we can obtain a sequence of constraint submanifolds 
$$
	\cdots\hookrightarrow M_f\hookrightarrow\cdots\hookrightarrow M_2\hookrightarrow M_1\hookrightarrow M\,.
$$
If this procedure stabilizes, that is, there exists a natural number $f\in\mathbb{N}$ such that $M_{f+1}=M_f$, and $\dim M_f>0$ we say that $M_f$ is the \textbf{final constraint submanifold}. In $M_f$ we can find solutions to equations \eqref{eq:Ham-eq-precocontact-vectorfields} which are tangent to $M_f$. Notice that, in the Lagrangian case, we need to impose the {\sc sode} condition.

In general, the Reeb vector fields are not tangent to the submanifolds provided by this constraint algorithm. One can continue the algorithm by demanding the tangency of the Reeb vector fields, as shown in \cite{DeLeo2019}. 
This is important if we want a Dirac--Jacobi bracket on the resulting constraint submanifold, which requires a Reeb vector field.

\subsection{Precocontact Lagrangian systems}

Consider a Lagrangian function $\L:\R\times \T Q\times\R\to \R$, and the associated objects $\eta_\L$ and $E_\L$ as defined in Section \ref{sec:cocontact-lagrangian-systems}.
\begin{dfn}
    A Lagrangian function $\L:\R\times \T Q\times\R\to \R$ is \textbf{admissible} if the Hessian of $\L$ has constant rank and $(\R\times \T Q\times\R,\d t,\eta_\L)$ is a precocontact manifold.
\end{dfn}
\begin{rmrk}\rm
Not every Lagrangian whose Hessian has constant rank leads to a precocontact structure. Consider the manifold $\R\times \T \R\times\R$, with coordinates $(t,q,v,s)$, and the Lagrangian $\L=vs$. The Hessian of $\L$ has constant rank equal to $0$, and the $1$-forms are $\tau=\d t$ and $\eta_\L=\d s-s\d q$. We have that $\d\eta_\L=-\d s\wedge\d q$, thus $\cl(\tau,\eta_\L) = 4 - 1=3$, which is odd (in a set where $s\neq0$) . The main problem of this structure is that there are no Reeb vector fields as defined in \ref{dfn:ReebSing}. \end{rmrk}
\begin{rmrk}\rm
The condition of being admissible is stronger than the condition proposed in \cite{DeLeo2019}, which is just for the Hessian to have constant rank. In light of the previous example, which can also be considered in the precontact setting, we believe it is necessary to require that the Lagrangian is admissible. Non-admissible Lagrangians will be studied in a future work.
\end{rmrk}

The precocontact Herglotz--Euler--Lagrange equations are defined as in Definition \ref{dfn:Euler-Lagrange-equations}, but there is no result about the solutions like Theorem \ref{thm:regular-lagrangian}. First of all, since we are dealing with a precocontact system, a constraint algorithm is required in order to find solutions. In particular, on the final constraint submanifold, the Herglotz--Euler--Lagrange equations do not depend on the choice of the Reeb vector fields, as proved in Corollary \ref{cor:independent-Reeb}. Moreover, the holonomy condition is not always recovered and it has to be imposed, leading to new constraints which have to be considered during the constraint algorithm.

\subsection{The canonical Hamiltonian formalism}

In the (hyper)regular case, the Legendre transform gives a diffeomorphism between $(\R\times\T Q\times\R, \d t, \eta_\L)$ and $(\R\times\cT Q\times\R, \tau = \d t, \eta)$ such that $\F\L^\ast\eta = \eta_\L$. For the singular case, the Legendre transform can be defined but, in general, $\P := \Ima(\F\L)=\F\L(\R\times\T Q\times\R)\varsubsetneq \R\times\cT Q\times\R$. Some regularity conditions are required to assure that a Hamiltonian precocontact system can be realised on $\P$.

\begin{dfn}
	A singular Lagrangian $\L$ is \textbf{almost-regular} if
	\begin{itemize}
	    \item $\L$ is admisible,
	    \item $\P := \Ima(\F\L)=\F\L(\R\times\T Q\times\R)$ is a closed submanifold of $\R\times\cT Q\times\R$,
	    \item the Legendre map $\F\L$ is a submersion onto its image,
	    \item the fibers $\F\L^{-1}(\F\L(t,v_q,s))\subset \R\times\T Q\times\R$ are connected submanifolds for every $(t,v_q,s)\in \R\times\T Q\times\R$.
	\end{itemize}
\end{dfn}

In the almost-regular case we can construct the triple$(\P, \tau, \eta_\P)$, where $\eta_\P = j_\P^\ast\F\L^\ast\eta\in\Omega^1(\P)$, $\tau=\F\L^\ast\d t$  and $j_\P\colon\P\hookrightarrow\R\times\cT Q\times\R$ is the natural embedding. Furthermore, the Lagrangian energy $E_\L$ is $\F\L$-projectable, that is, there is a unique $H_\P\in\Cinfty(\P)$ such that $E_\L = (j_\P\circ\F\L)^\ast H_\P$. Then, $(\P,\tau, \eta_\P,H_\P) $ is a precocontact Hamiltonian system.

The contact Hamiltonian equations for $X_\P\in\X(\P)$ are \eqref{eq:Ham-eq-cocontact-vectorfields} adapted to this situation. As in the Lagrangian formalism, these equations are not necessarily consistent everywhere on $\P$ and we must implement a suitable \textit{constraint algorithm} in order to find a \textit{final constraint submanifold} $\P_f\hookrightarrow\P$ (if it exists) where there exist vector fields $X\in\X(\P)$, tangent to $\P_f$, which are (not necessarily unique) solutions to \eqref{eq:Ham-eq-cocontact-vectorfields} on $\P_f$.

\section{Damped mechanical systems with holonomic constraints}

There are two classes of constraints: \textbf{holonomic}, which only depend on the generalized coordinates and time, and \textbf{nonholonomic}, which have a dependence on velocities or momenta. Nonholonomic constraints are more intricate and there are several ways to implement them (see \cite{Gra2003} for a general study of these kind of constraints in the symplectic setting).

In a recent article \cite{DeLeo2021d} the authors consider contact systems with nonholonomic constraints and show how a contact system can be understood as a symplectic system with nonholonomic constraints. This chapter adds to these results by considering singular Lagrangians and allowing explicit dependence on time of both the Lagrangian and the constraints. On the other hand, we only consider holonomic constraints.

Consider a Lagrangian precocontact system $(\R\times\T Q\times\R,{\L'})$, with $\tau=\d t$. We will denote by $\eta_{\L'}$, $\C'$ and $E_{\L'}$ the contact form, characteristic distribution and Lagrangian energy associated to the system $(\R\times\T Q\times \R,\L')$. Now we add a set of independent constraints $f^\alpha(t,q^i,s)=0$ with $1\leq\alpha\leq d$, which define a submanifold $S\hookrightarrow \R\times\T Q\times\R$. In order to define a suitable constraint submanifold, these constraints functions must verify the condition
$$ \rk\left(\dparder{f^\alpha}{q^i}\right) = d\,,\quad \forall t,s\,. $$
Therefore, we can understand these constraints as a constraint in the generalized coordinates $(q^i, s)$ for every value of time $t$.

In order to find the precocontact Lagrangian vector field for $(S,\L)$, we add $d$ new configuration variables $\lambda_\alpha$, thus the enlarged manifold is $\R\times\T (Q\times\mathbb{R}^d)\times\R$, with local variables $(t,q^i,v^i,\lambda_\alpha,v_\alpha,s)$. We consider the new Lagrangian:
\[
\L=\L'+\lambda_\alpha f^\alpha\,.
\]
thus $\lambda_\alpha$ act as Lagrange multipliers and we consider them dynamical variables.

We have the precocontact Lagrangian system $(\R\times\T (Q\times\mathbb{R}^d)\times\R,\L)$. The contact form is the same $\eta_\L=\eta_{\L'}\,$, and the characteristic distribution is
\[
\C=\C'\cup\left<\frac{\partial}{\partial \lambda_\alpha},\frac{\partial}{\partial v_\alpha}\right>\,.
\]
And the Lagrangian energy  is
\[
E_{\L}=E_{\L'}-\lambda_\alpha f^\alpha\,.
\]
The primary constraints are generated by sections of $\C$, which we can separate in two classes, those which are sections of $\C'$ and $\dparder{}{\lambda_\alpha}$ and $\dparder{}{v_\alpha}$. The last ones give the constraints:
\[
    \Lie_{\tparder{}{\lambda_\alpha}}E_\L=-f^\alpha=0\,;\quad \Lie_{\tparder{}{v_\alpha}}E_\L =0\,.
\]
Thus, we recover the original constraints from the constraint algorithm. If $Y$ is a section of $\C'$, then
\[
    \Lie_YE_{\L}=\Lie_YE_{\L'}-\lambda_\alpha \Lie_Yf^\alpha=0\,.
\]
The primary constraints for the Lagrangian $\L'$ (that is, without imposing $f^\alpha=0$) are $\Lie_YE_{\L'}$. Thus, these constraints become coupled with $f^\alpha$ and are not conserved in general. From this point the constraints algorithm should continue imposing the tangency condition. The outcome will depend on the particular Lagrangian and constraints.

We will now compute the dynamical equations, which are complementary to the constraint algorithm. Consider a holonomic vector field:
$$ X = f\parder{}{t} + v^i\parder{}{x^i} + v_\alpha\parder{}{\lambda_\alpha} + G^i\parder{}{v^i} + G_\alpha\parder{}{v_\alpha} + g\parder{}{s}\,. $$
Then the precocontact equations \eqref{eq:Ham-eq-precocontact-vectorfields} are
\begin{align*}
    f &= 1\,,\\
    g &= \L\,,\\
    f^\alpha&=0\,,\\
    \parderr{\L'}{t}{v^i} + v^j\parderr{\L'}{q^j}{v^i} + G^j\parderr{\L'}{v^j}{v^i} + \L'\parderr{\L'}{s}{v^i} - \parder{\L'}{q^i}&=\parder{\L'}{s}\parder{\L'}{v^i}+\lambda_\alpha\left(\parder{f^\alpha}{q^i}+\parder{f^\alpha}{s}\parder{\L'}{v^i}\right)\,.
\end{align*}

\section{Examples}

In the following examples we consider mechanical systems in Riemannian manifolds described by Lagrangians of the form $L = K - U$, with $K$ and $U$ are the kinetic and potential energy respectively \cite{Abr1978}. These Lagrangians are hyperregular. In order to introduce dissipation and external forces, we will add some additional terms to these Lagrangians such as $\gamma s$ and $G(t, q)$ respectively. We also consider systems subjected to holonomic time-dependent constraints $f^\alpha(t,q)$.

\subsection{Damped forced harmonic oscillator}

In this example we are going to study a forced harmonic oscillator with damping. We will develop both the Lagrangian and Hamiltonian formalisms. Consider a harmonic oscillator of mass $m$ with elastic constant $k$ and an external force $f(t)$ depending on time.

\subsubsection*{Lagrangian formalism}

The configuration manifold for this system is $Q=\R$ equipped with coordinate $q$. Consider the phase manifold $\R\times\T Q\times\R$ equipped with coordinates $(t, q, v, s)$ and the Lagrangian function $\L\colon\R\times\T Q\times\R\to\R$ given by
\begin{equation}\label{eq:Lagrangian-damped-forced-oscillator}
    \L(t,q,v,s) = \frac{1}{2}mv^2 - \frac{k}{2}q^2 + q f(t) - \frac{\gamma}{m} s\,,
\end{equation}
where $f(t)$ is a time-dependent external force. The Lagrangian energy associated to this Lagrangian function is
$$ E_\L = \frac{1}{2}mv^2 + \frac{k}{2}q^2 - qf(t) + \frac{\gamma}{m} s\,, $$
and its differential is
$$ \d E_\L = -qf'(t)\d t + (kq - f(t))\d q + mv\d v + \frac{\gamma}{m}\d s\,. $$

\begin{figure}
    \centering
    \begin{subfigure}[h]{0.48\textwidth}
        \centering
        \includegraphics[width=\textwidth]{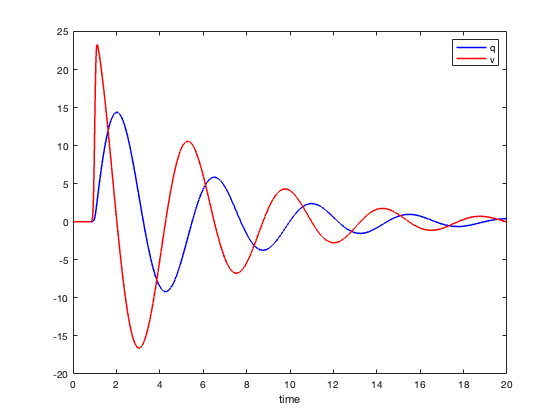}
        \caption{Position and velocity of with respect to time}
        \label{fig:oscillator-pos-vel}
    \end{subfigure}
    \hfill
    \begin{subfigure}[h]{0.48\textwidth}
        \centering
        \includegraphics[width=\textwidth]{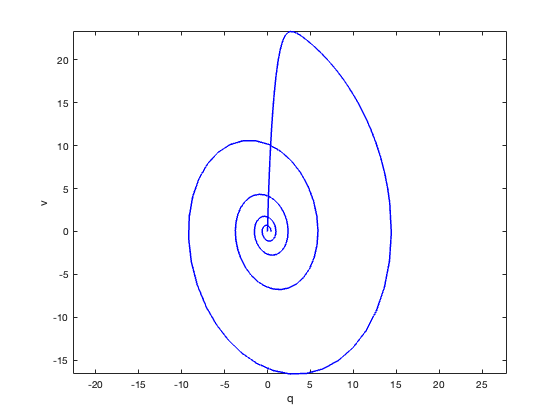}
        \caption{Phase portrait of the oscillator}
        \label{fig:oscillator-phase}
    \end{subfigure}
    \caption{These figures depict the evolution of the position and velocity of the damped harmonic oscillator with respect to time and the phase portrait of the system.}
    \label{fig:oscillator}
\end{figure}

The Cartan 1-form for the Lagrangian $\L$ is $ \theta_\L = mv\d q$. The contact 1-form is $\eta_\L = \d s - mv\d q$, and its differential is $\d\eta_\L = m\d q\wedge\d v$. The Reeb vector fields are:
$$ R_t^\L = \parder{}{t}\ ,\quad R_s^\L = \parder{}{s}\,. $$
Consider a vector field $X\in\X(\R\times\T Q\times\R)$ with local expression
$$ X = a\parder{}{t} + F\parder{}{q} + G\parder{}{v} + g\parder{}{s}\,. $$
The cocontact Lagrangian equations \eqref{eq:Euler-Lagrange-field} for this vector field yield the conditions
\begin{equation*}
    \begin{dcases}
        F = m\,,\\
        G = -\frac{k}{m}q + \frac{f(t)}{m} - \frac{\gamma}{m} v\,,\\
        g = \L\,,\\
        a = 1\,.
    \end{dcases}
\end{equation*}
Hence, the vector field $X$ is
$$ X = \parder{}{t} + v\parder{}{q} + \left( -\frac{k}{m}q + \frac{f(t)}{m} - \gamma v \right)\parder{}{v} + \L\parder{}{s}\,. $$
Its integral curves $(t(r),q(r),v(r),s(r))$ satisfy the system of differential equations
\begin{equation*}
    \begin{dcases}
        \dot t = 1\,,\\
        m\ddot q + \gamma \dot q + kq = f(t)\,,\\
        \dot s = \L\,.
    \end{dcases}
\end{equation*}

In Figure \ref{fig:oscillator-pos-vel} we see the evolution with respect to time of the position and the velocity of the damped oscillator taking as external force a smooth pulse at $t = 1$. We can see the damping of the position and the velocity. In Figure \ref{fig:oscillator-phase} we have represented the phase portrait of the same solution where we can see the initial pulse and how the system decays to the equilibrium point due to the friction.

In Figure \ref{fig:oscillator-energy} we can see the dissipation of both the Lagrangian energy and the mechanical energy given by
$$ E_m = \frac{1}{2}m v^2 + \frac{1}{2}kq^2\,. $$
Notice that the Lagrangian energy decays exponentially, while the mechanical energy follows the evolution of the Lagrangian energy but oscillating around it.

\begin{figure}[h]
    \centering
    \includegraphics[width=0.7\textwidth]{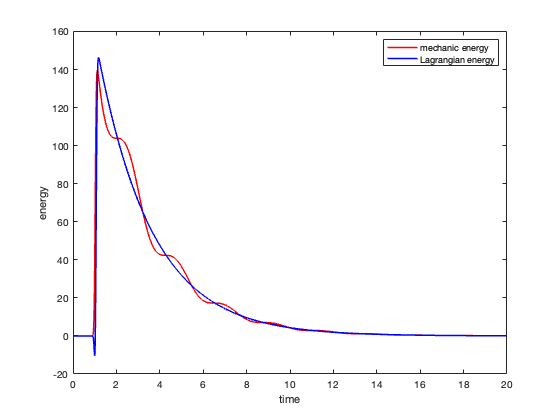}
    \caption{Evolution of the mechanical energy (red) and the Lagrangian energy (blue)}
    \label{fig:oscillator-energy}
\end{figure}

\subsubsection*{Hamiltonian formalism}

Consider the Legendre map associated to the Lagrangian function \eqref{eq:Lagrangian-damped-forced-oscillator}:
$$ \F\L\colon\R\times\T Q\times\R\to\R\times\cT Q\times\R\,,$$
which is given by
$$ \F\L(t,q,v,s) = \left(t,q,p\equiv mv,s\right)\,. $$
Notice that the Legendre map $\F\L$ is a global diffeomorphism and hence $\L$ is an hyperregular Lagrangian.

Then, $\R\times\cT Q\times\R$ is equipped with the cocontact structure $(\d t, \d s - p\d q)$. The Hamiltonian function $H$ such that $\F\L^\ast H = E_\L$ is
$$ H(t,q,p,s) = \frac{p^2}{2m} + \frac{k}{2}q^2 - qf(t) + \frac{\gamma}{m} s\,. $$
Then, a vector field $Y\in\X(\R\times\cT Q\times\R)$ is a solution to Hamilton's equation \eqref{eq:Ham-eq-cocontact-vectorfields} if it has local expression
\begin{equation*}
    Y = \parder{}{t} + \frac{p}{m}\parder{}{q} + \left( -kq + f(t) - \frac{p}{m}\gamma \right)\parder{}{p} + \left( \frac{p^2}{2m} - \frac{k}{2}q^2 + qf(t) - \frac{\gamma}{m} s \right)\parder{}{s}\,.
\end{equation*}
Its integral curves $(t(r),q(r),p(r),s(r))$ satisfy
\begin{equation*}
    \begin{dcases}
        \dot t = 1\,,\\
        \dot q = \frac{p}{m}\,,\\
        \dot p = -kq + f(t) - \frac{p}{m}\gamma\,,\\
        \dot s = \frac{p^2}{2m} - \frac{k}{2}q^2 + qf(t) - \frac{\gamma}{m} s\,.
    \end{dcases}
\end{equation*}
Combining the second and the third equations above, we obtain the second-order differential equation
$$ m\ddot q + \gamma \dot q + kq = f(t) \,. $$

\subsection{A time-dependent system with central force and friction}

Consider the Kepler problem in the case where the mass of the particle subjected to the central force is a non-vanishing function of time $m(t)$. It is clear that the motion of the particle is on a plane and hence the configuration manifold is $Q=\R^2$ endowed with coordinates $(r,\varphi)$.

The phase bundle $\R\times\cT Q\times\R$ with coordinates $(t, r, \varphi, p_r, p_\varphi, s)$ has a natural cocontact structure given by the 1-forms $\tau = \d t$ and $\eta = \d s - p_r\d r - p_\varphi\d\varphi$. The Reeb vector fields are
$$ R_t = \parder{}{t}\ ,\quad R_s = \parder{}{s}\,. $$
Consider the Hamiltonian function $H\in\Cinfty(\R\times\cT Q\times\R)$ given by
$$ H(t, r, \varphi, p_r, p_\varphi, s) = \frac{p_r^2}{2m(t)} + \frac{p_\varphi^2}{2m(t)r^2} + \frac{k}{r} + \gamma s\,. $$
The vector field $X\in\X(\R\times\cT Q\times\R)$ satisfying equations \eqref{eq:Ham-eq-cocontact-vectorfields} has local expression
\begin{multline*}
    X = \parder{}{t} + \frac{p_r}{m(t)}\parder{}{r} + \frac{p_\varphi}{m(t)r^2}\parder{}{\varphi} + \left( \frac{p_\varphi^2}{m(t)r^3} + \frac{k}{r^2} - \gamma p_r \right)\parder{}{p_r}\\
    - \gamma p_\varphi\parder{}{p_\varphi} + \left( \frac{p_r^2}{2m(t)} + \frac{p_\varphi^2}{2m(t)r^2} - \frac{k}{r} - \gamma s \right)\parder{}{s}\,.
\end{multline*}
Then, the integral curves $(t, r, \varphi, p_r, p_\varphi, s)$ satisfy
\begin{equation*}
    \begin{dcases}
        \dot t = 1\,,\\
        m(t)\dot r = p_r\,,\\
        m(t)r^2\dot\varphi = p_\varphi\,,\\
        \dot p_r = \frac{p_\varphi^2}{m(t)r^3} + \frac{k}{r^2} - \gamma p_r\,,\\
        \dot p_\varphi = -\gamma p_\varphi\,,\\
        \dot s = \frac{p_r^2}{2m(t)} + \frac{p_\varphi^2}{2m(t)r^2} - \frac{k}{r} - \gamma s\,.
    \end{dcases}
\end{equation*}
Hence, the integral curves must fulfill the system of second-order equations
\begin{equation*}
    \begin{dcases}
        \frac{\d}{\d t}\left( m(t)\dot r \right) = m(t)r\dot\varphi^2 + \frac{k}{r^2} - \gamma m(t)\dot r\,,\\
        \frac{\d}{\d t}\left( m(t)r^2\dot\varphi \right) = -\gamma m(t)r^2\dot\varphi\,.
    \end{dcases}
\end{equation*}

\subsection{Damped pendulum with variable length}

Consider a damped pendulum of mass $m$ with variable length $\ell(t)$ \cite{Gra2005}. Its position in the plane can be described using polar coordinates $(r,\theta)$. The constraint $r = \ell(t)$ will be introduced in the Lagrangian function via a Lagrange multiplier. The phase space of this system is the bundle $\R\times\T\R^3\times\R$, equipped with coordinates $(t, r, \theta, \lambda, v_r, v_\theta, v_\lambda, s)$. The Lagrangian function describing this system is
$$ \L = \frac{1}{2}m(v_r^2 + r^2v_\theta^2) - mgr(1-\cos\theta) + \lambda(r-\ell(t)) - \gamma s\in\Cinfty(\R\times\T\R^3\times\R)\,, $$
where $\lambda$ is the Lagrange multiplier.

The contact 1-form is
$$ \eta_\L = \d s - mv_r\d r - mr^2v_\theta\d\theta\,, $$
then
$$ \d\eta_\L = m\d r\wedge\d v_r + 2mrv_\theta\d\theta\wedge\d r + mr^2\d\theta\wedge\d v_\theta\,, $$
and we have the 1-form $\tau = \d t$. Hence, $(\R\times\T\R^3\times\R, \tau,\eta_\L)$ is a precocontact manifold. We can take as Reeb vector fields $\Reeb_s = \tparder{}{s}$, $\Reeb_t = \tparder{}{t}$. The characteristic distribution of $(\tau,\eta_\L)$ is
$$ \C = \ker\tau\cap\ker\eta_\L\cap\ker\d\eta_\L = \left\langle \parder{}{\lambda},\parder{}{v_\lambda} \right\rangle\,. $$

The Lagrangian energy associated to $\L$ is
$$ E_\L = \frac{1}{2}m(v_r^2 + r^2v_\theta^2) + mgr(1-\cos\theta) - \lambda(r-\ell(t)) + \gamma s\,, $$
and thus
\begin{align*}
    \d E_\L - \Reeb_t(E_\L)\d t - \Reeb_s(E_\L)\eta &= mv_r\d v_r + mr^2v_\theta\d v_\theta + \big(mrv_\theta^2 + mg(1-\cos\theta) - \lambda + \gamma m v_r\big)\d r \\&\quad + \big( mgr\sin\theta + \gamma mr^2 v_\theta \big)\d\theta - (r - \ell(t))\d\lambda\,.
\end{align*}

Consider a \textsc{sode} $X\in\X(\R\times\T\R^3\times\R)$ with local expression
$$ X = f\parder{}{t} + v_r\parder{}{r} + v_\theta\parder{}{\theta} + v_\lambda\parder{}{\lambda} + G_r\parder{}{v_r} + G_\theta\parder{}{v_\theta} + G_\lambda\parder{}{v_\lambda} + g\parder{}{s}\,. $$

\begin{figure}[h]
    \centering
    \begin{subfigure}{0.48\textwidth}
        \centering
        \includegraphics[width=\textwidth]{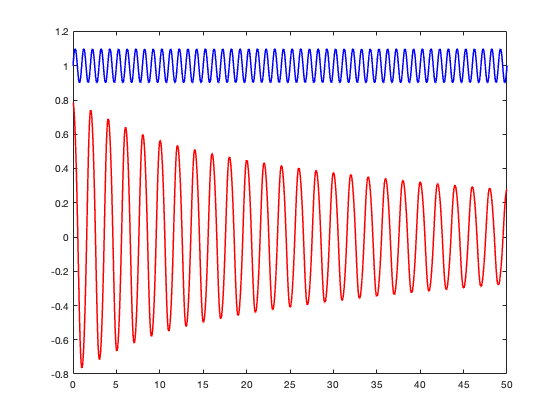}
        \caption{Radius (blue) and angle $\theta$ (red) of with respect to time}
        \label{fig:pendulum1-g05}
    \end{subfigure}
    \hfill
    \begin{subfigure}{0.48\textwidth}
        \centering
        \includegraphics[width=\textwidth]{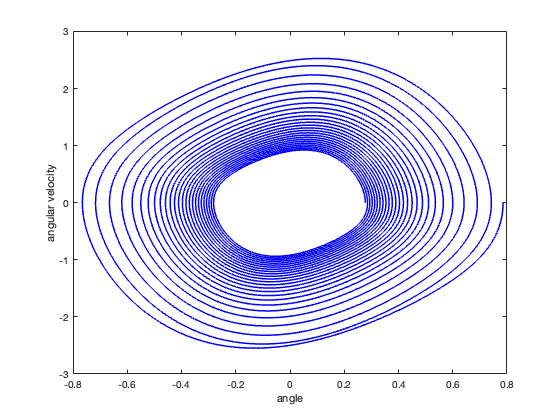}
        \caption{Phase portrait of the pendulum ($\theta$ and $\dot\theta$)}
        \label{fig:pendulum2-g05}
    \end{subfigure}
    \hfill
    \begin{subfigure}{0.48\textwidth}
        \centering
        \includegraphics[width=\textwidth]{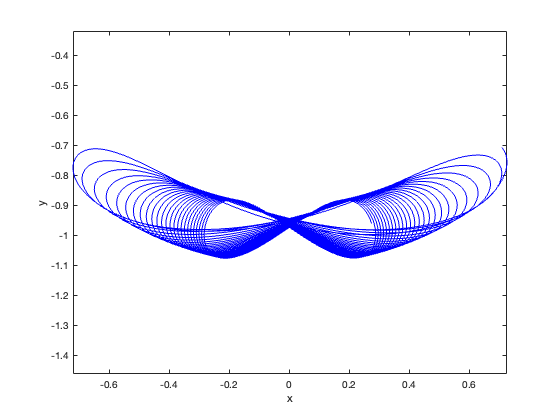}
        \caption{Trajectory of the pendulum in the plane $XY$}
        \label{fig:pendulum3-g05}
    \end{subfigure}
    \hfill
    \begin{subfigure}{0.48\textwidth}
        \centering
        \includegraphics[width=\textwidth]{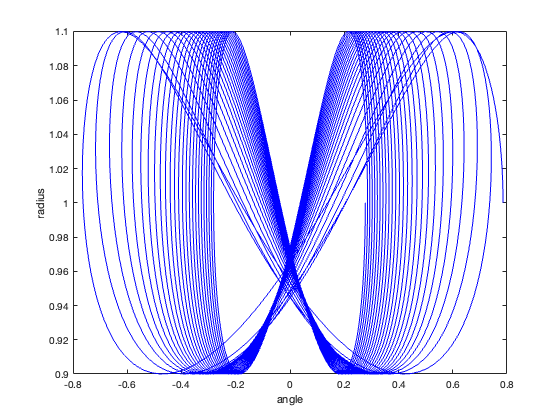}
        \caption{Radius with respect to the angle $\theta$}
        \label{fig:pendulum4-g05}
    \end{subfigure}
    \caption{Plots of the pendulum with friction coefficient $\gamma=0.5$ and $\ell(t) = 1 + 0.1\sin(2\pi t)$}
    \label{fig:pendulum-g05}
\end{figure}

The dynamical equations for the vector field $X$ yield the conditions
\begin{gather*}
    f = 1\,,\quad v_r = v_r\,,\quad v_\theta = v_\theta\,,\quad r - \ell(t) = 0\,,\quad g = \L\,,\\
    G_r = rv_\theta^2 - g(1-\cos\theta) + \frac{\lambda}{m} - \gamma v_r\,,\quad G_\theta = -\frac{2}{r}v_r v_\theta - \frac{g}{r}\sin\theta - \gamma v_\theta
\end{gather*}
and we obtain the constraint function
$$ \xi_1 = r - \ell(t)\,, $$
defining the first constraint submanifold $M_1\hookrightarrow \R\times\T\R^3\times\R$.

Hence, the vector field $X$ has the form
\begin{align*}
    X &= \parder{}{t} + v_r\parder{}{r} + v_\theta\parder{}{\theta} + v_\lambda\parder{}{\lambda} + \left(rv_\theta^2 - g(1-\cos\theta) + \frac{\lambda}{m} - \gamma v_r\right)\parder{}{v_r}\\
    &\quad + \left( -\frac{2}{r}v_r v_\theta - \frac{g}{r}\sin\theta - \gamma v_\theta \right)\parder{}{v_\theta} + G_\lambda\parder{}{v_\lambda} + \L\parder{}{s}\qquad \text{(on $M_1$)}\,.
\end{align*}

\begin{figure}[h]
    \centering
    \begin{subfigure}{0.48\textwidth}
        \centering
        \includegraphics[width=\textwidth]{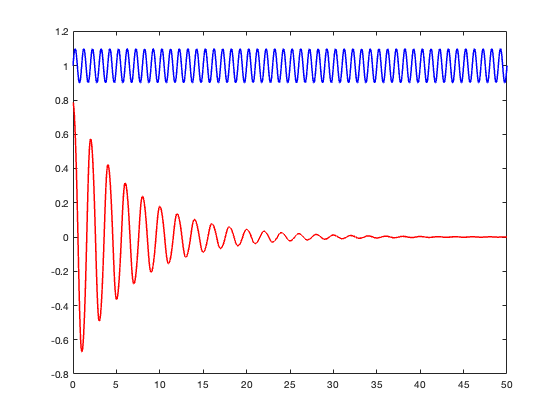}
        \caption{Radius (blue) and angle $\theta$ (red) with respect to time}
        \label{fig:pendulum1-g075}
    \end{subfigure}
    \hfill
    \begin{subfigure}{0.48\textwidth}
        \centering
        \includegraphics[width=\textwidth]{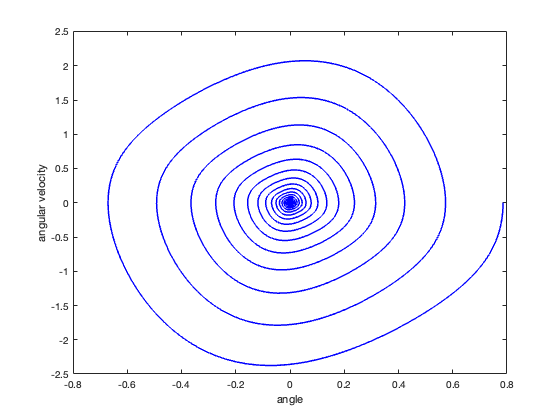}
        \caption{Phase portrait of the pendulum ($\theta$ and $\dot\theta$)}
        \label{fig:pendulum2-g075}
    \end{subfigure}
    \hfill
    \begin{subfigure}{0.48\textwidth}
        \centering
        \includegraphics[width=\textwidth]{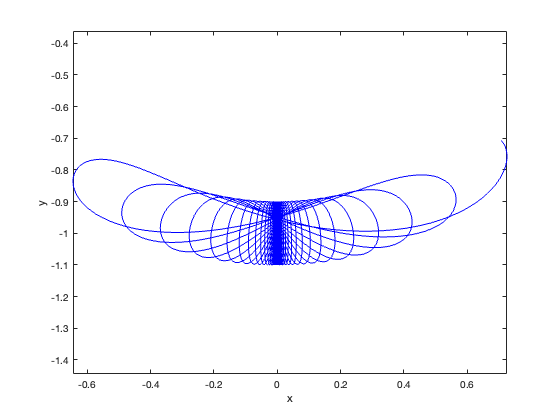}
        \caption{Trajectory of the pendulum in the plane $XY$}
        \label{fig:pendulum3-g075}
    \end{subfigure}
    \hfill
    \begin{subfigure}{0.48\textwidth}
        \centering
        \includegraphics[width=\textwidth]{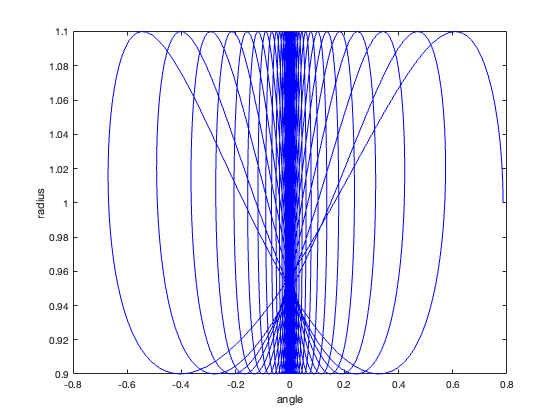}
        \caption{Radius with respect to the angle $\theta$}
        \label{fig:pendulum4-g075}
    \end{subfigure}
    \caption{Plots of the pendulum with friction coefficient $\gamma=0.75$ and $\ell(t) = 1 + 0.1\sin(2\pi t)$}
    \label{fig:pendulum-g075}
\end{figure}

Imposing the tangency of the vector field $X$ to the submanifold $M_1$, namely the condition $\Lie_X \xi_1 = 0$, we obtain the constraint
$$ \xi_2 = v_r - \ell'(t) = 0\qquad \text{(on $M_1$)}\,, $$
defining a new constraint submanifold $M_2\hookrightarrow M_1$. The tangency condition $\Lie_X \xi_2 = 0$ of the vector field $X$ to the submanifold $M_2$ yields the new constraint function
$$ \xi_3 = mr v_\theta^2 - g(1-\cos\theta) + \lambda - m\gamma \ell'(t) - \ell''(t)\qquad \text{(on $M_2$)}\,, $$
defining a new constraint submanifold $M_3\hookrightarrow M_2$, and we also get $G_r = \ell''(t)$. Imposing again the tangency condition we obtain a new constraint function
$$ \xi_4 = \Lie_X \xi_3 = v_\lambda - 3m\ell'(t) v_\theta^2 - 2m\ell(t)\gamma v_\theta^2 - mg v_\theta\sin\theta - m\gamma\ell''(t) - m\ell'''(t)\qquad \text{(on $M_3$)}\,, $$
defining the submanifold $M_4\hookrightarrow M_3$. Requiring $X$ to be tangent to $M_4$ we determine the last coefficient $G_\lambda$, whose expression we will omit, and no new constraints appear. Thus, there is a unique vector field solution to equations \eqref{eq:Euler-Lagrange-field} and has local expression
\begin{align*}
    X &= \parder{}{t} + \ell'(t)\parder{}{r} + v_\theta\parder{}{\theta} + \left( 3m\ell'(t) v_\theta^2 + 2m\ell(t)\gamma v_\theta^2 + mg v_\theta\sin\theta + m\gamma\ell''(t) + m\ell'''(t) \right)\parder{}{\lambda} \\
    &\quad + \ell''(t)\parder{}{v_r} + \left( -\frac{2}{r}v_r v_\theta - \frac{g}{r}\sin\theta - \gamma v_\theta \right)\parder{}{v_\theta} + G_\lambda\parder{}{v_\lambda} + \L\parder{}{s}\qquad \text{(on $M_4$)}\,.
\end{align*}

The integral curves of this vector field satisfy the following second-order differential equation
\begin{equation}\label{eq:damped-pendulum}
    \ddot\theta = -\gamma\dot\theta - 2\frac{\ell'(t)}{\ell(t)}\dot\theta - \frac{g}{\ell(t)}\sin\theta \,.
\end{equation}
Notice that if we consider a pendulum with fixed length $\ell(t) = \ell_\circ$, we recover the usual equation for a damped pendulum:
$$ \ddot\theta = -\gamma\dot\theta - \frac{g}{\ell_\circ}\sin\theta \,. $$
On the other hand, setting $\gamma=0$ in equation \eqref{eq:damped-pendulum}, we obtain the equation of the simple pendulum with variable length studied in \cite{Gra2005}.

In Figures \ref{fig:pendulum-g05} and \ref{fig:pendulum-g075} we have represented a couple of simulations of a damped pendulum of mass $m=1$ with variable length considering $\ell(t) = 1 + 0.1\sin(2\pi t)$, friction coefficients $\gamma = 0.5$ and $\gamma = 0.75$ respectively, and initial conditions $\theta(0) = \pi/4$ and $\dot\theta(0) = 0$. We have plotted the evolution of the radial and angular coordinates and we can see the loss of amplitude, and hence of energy, of the system. Notice that in this example, the energy does not tend to zero, but to a positive constant since the radial coordinates keeps oscillating forever.

\section{Conclusions and further research}

In this paper we have introduced a geometrical formulation for time-dependent contact systems by defining a new geometric structure: cocontact manifolds. This new notion combines the well-known contact and cosymplectic structures. We have also proved that cocontact manifolds are Jacobi manifolds and defined and characterize the notions of isotropic and Legendrian submanifolds.

This geometrical setting allows us to develop the Hamiltonian and Lagrangian formalisms for time-dependent contact systems, generalizing those for contact systems \cite{Bra2017b,DeLeo2019b,Gas2019} and cosymplectic systems \cite{DeLeo2017}. In addition, we have studied the problem where the system is defined by a singular Lagrangian, thus introducing the notion of precocontact structure. This is useful since many systems are defined by singular Lagrangians. As an application, we have studied the particular case of cocontact systems with time-dependent holonomic constraints.

We have worked out two regular examples: the damped forced harmonic oscillator and the Kepler problem with non-constant mass and friction; and a singular one: a pendulum with variable length and friction. This last example is singular because we have introduced the constraint with a Lagrange multiplier and the constraint algorithm gives back the constraint. Computer simulations of some of these examples have been included.

The structures introduced in this paper could be used to improve our understanding of time-dependent dissipative systems. For instance, providing new geometric integrators \cite{Ver2019, Bra2020b, Zadra2020, AnahorySimoes2021} from the discretization of the obtained equations, discussing symmetries and their associated dissipated and conserved quantities, and studying reduction procedures such as coisotropic reduction \cite{Abr1978,DeLeo2019b} and Marsden--Weinstein reduction \cite{Mar1974}. It would be also interesting to state the Hamilton--Jacobi theory for these systems and describe the Skinner--Rusk unified formalism for cocontact systems.

The formulation presented in this work is also a first step towards finding a geometric formalism for non-autonomous dissipative field theories based on the $k$-contact setting \cite{Gas2020,Gas2021,Gra2021} and generalizing the multisymplectic formalism \cite{Car1991,Rom2009}. The $k$-contact formalism allows to describe autonomous field theories, such as field theories with damping, some equations from circuit theory, such as the so-called telegrapher's equation, or the Burgers' equation. Nevertheless, there are many examples of non-autonomous field theories, like Maxwell's equations with a non constant charge density or general relativity with matter sources that require a formulation for non-autonomous field theories.

\section{Acknowledgements}

We acknowledge fruitful discussions and comments from our colleague Narciso Román-Roy. MdL acknowledges the financial support of the Ministerio de Ciencia e Innovación (Spain), under grants PID2019-106715GB-C2, ``Severo Ochoa Programme for Centres of Excellence in R\&D'' (CEX2019-000904-S) and EIN2020-112197, funded by AEI/10.13039/501100011033 and European Union NextGenerationEU/PRTR. JG, XG, MCML and XR acknowledge the financial support of the Ministerio de Ciencia, Innovaci\'on y Universidades (Spain), project PGC2018-098265-B-C33. 

\bibliographystyle{abbrv}
\bibliography{bibliografia.bib}

\end{document}